\documentclass[twocolumn,pra,epsfig,showpacs,amsmath,amsfonts,amssymb,longbibliography]{revtex4-1}

 \usepackage{graphicx}
 \usepackage{color}
 \usepackage{xcolor}
 \usepackage{subfigure}
 \usepackage{verbatim}
 \usepackage{bbm}
 \usepackage{amsthm} 
 \newtheorem{thm}{Theorem}
 \newtheorem{cor}{Corollary}[thm]
 \newtheorem{lem}{Lemma}

 \newcommand{\be}{\begin{eqnarray}}
 \newcommand{\ee}{\end{eqnarray}}
 \newcommand{\ra}{\rangle}
 \newcommand{\la}{\langle}

 \newcommand{\wt}{\widetilde}

\begin{document}

\title{Markovian and non-Markovian quantum measurements}

\author{Jennifer R. Glick}
\email{patte399@msu.edu}

\author{Christoph Adami}
\email{adami@msu.edu}

\affiliation{Department of Physics and Astronomy, Michigan State University, East Lansing, Michigan 48824, USA}

\date{\today}

\begin{abstract}Consecutive measurements performed on the same quantum system can reveal fundamental insights into quantum theory's causal structure, and probe different aspects of the quantum measurement problem. According to the Copenhagen interpretation, measurements affect the quantum system in such a way that the quantum superposition collapses after the measurement, erasing any knowledge of the prior state. We show here that counter to this view,  unamplified measurements (measurements where all variables comprising a pointer are controllable) have coherent ancilla density matrices that encode the memory of the entire set of quantum measurements, and that the quantum chain of a set of consecutive unamplified measurements is non-Markovian. In contrast, sequences of amplified measurements (measurements where at least one pointer variable has been lost) are equivalent to a quantum Markov chain. An analysis of arbitrary non-Markovian quantum chains of measurements reveals that all of the information necessary to reconstruct the chain is encoded on its boundary (the state preparation and the final measurement), reminiscent of the holographic principle.

\end{abstract}

\maketitle

\numberwithin{equation}{section}

 \section{Introduction}
 The physics of consecutive (sequential) measurements on the same quantum system has enjoyed increased attention as of late, as it probes the causal structure of quantum mechanics~\cite{Brukner2014}. It is of interest to researchers concerned about the apparent lack of time-reversal invariance of Born's rule~\cite{Rovelli2016,OreshkovCerf2015}, as well as to those developing a consistent formulation of covariant quantum mechanics~\cite{ReisenbergerRovelli2002,OlsonDowling2007}, which does not allow for a time variable to define the order of (possibly non-commuting) projections~\cite{OreshkovCerf2016}.

 Consecutive measurements can be seen to challenge our understanding of quantum theory in an altogether different manner, however. According to standard theory, a measurement causes the state of a quantum system to ``collapse'', repreparing it as an eigenstate of the measured operator so that after multiple consecutive measurements on the quantum system any information about the initial preparation is erased. However, recent investigations of sequential measurements on a single quantum system with the purpose of optimal state discrimination have already hinted that quantum information survives the collapse~\cite{Nagalietal2012,Bergouetal2013}, and that information about a chain of sequential measurements can be retrieved from the final quantum state~\cite{HilleryKoch2016}. 
 
 Here we investigate the circumstances that make chains of quantum measurements ``Markovian''---meaning that each consecutive measurement ``wipes the slate clean'' so that retrodiction of quantum states~\cite{HilleryKoch2016} is impossible---and under what conditions the quantum trajectory remains coherent so that the memory of previous measurements is preserved.

In particular, we study the {\em relative state} of measurement devices (both quantum and classical) in terms of quantum information theory, to ascertain how much information about the quantum state appears in the measurement devices, and how this information is distributed. We find that a crucial distinction refers to the ``amplifiability" of a quantum measurement, that is, whether a result is encoded in the states of a closed or an open system, and conclude that a unitary relative-state description makes predictions that are different from a formalism that assumes quantum state reduction.

 While the suggestion that the relative state description of quantum measurement~\cite{Everett1957} (see also ~\cite{Zeh1973,Deutsch1984,CerfAdami1996,CerfAdami1998,GlickAdami-bell-2017}) and the Copenhagen interpretation are at odds and may lead to measurable differences has been made before~\cite{Zeh1973,Deutsch1984}, here we frame the problem of consecutive measurements in the language of quantum information theory, which allows us to make these differences manifest.

 We begin by outlining in Sec.~\ref{sec:two} the unitary description of quantum measurement discussed previously~\cite{CerfAdami1996,CerfAdami1998,GlickAdami-bell-2017}, and apply it in Sec.~\ref{sec:three} to a sequence of quantum measurements where the pointer---meaning a set of quantum ancilla states---remains under full control of the experimenter. In such a closed system, the pointer can in principle decohere if it is composed of more than one qubit, but this decoherence can be reversed in general. We prove in Theorems~\ref{theorem-cq-prepared} and~\ref{theorem-cq-unprepared} properties of the entropy of a chain of consecutive measurements that imply that the entropy of such chains resides in the last (or first and last) measurements. We then show that for coherence to be preserved in such chains, measurements cannot be arbitrarily amplified---in contrast to the macroscopic measurement devices that are necessarily open systems.

 In Sec.~\ref{sec:markov}, we analyze sequences of amplifiable---that is, macroscopic---measurements and prove in Theorem~\ref{theorem-markov} that amplified measurement sequences are Markovian. Corollary~\ref{corollary-markov} asserts an information-theoretic statement of the general idea that two macroscopic measurements anywhere on a Markov chain must be uncorrelated given the state of all the measurement devices that separate them in the chain. This corollary epitomizes the essence of the Copenhagen idea of quantum state reduction in terms of the conditional independence of measurement devices that are not immediately in each other's past or future. It is consistent with the notion that the measurements collapsed the state of the wavefunction, erasing any conditional information that a detector could have had about prior measurements. However, no irreversible reduction occurs and all amplitudes in the underlying pure-state wavefunction continue to evolve unitarily.

 Section \ref{sec:five} unifies the two previous sections by proving three statements (Theorems~\ref{theorem-last-info},~\ref{theorem-past-info}, and~\ref{theorem-chain}) that relate information-theoretic quantities pertaining to unamplified measurements to the corresponding expressions for amplified measurements. We show that, in general, amplification leads to a loss of information.

 After a brief application of the collected concepts and results to standards such as quantum state preparation, the double-slit experiment and the Zeno effect, in Sec.~\ref{sec:six}, we close with conclusions.

 \section{Theory of Quantum Measurement}
 \label{sec:two}

 \subsection{The measurement process}
 \label{sec:prepared-states}
 
 Suppose a given quantum system is in the initial state
 \begin{equation} \label{Q-prepared}
     |Q\ra = \sum_{x_1=1}^d \alpha^{(1)}_{x_1} \, |\wt{x}_1\ra,
 \end{equation}
 where $\alpha^{(1)}_{x_1}$ are complex amplitudes. Here, $Q$ is expressed in terms of the $d$ orthonormal basis states $|\wt{x}_1\ra$ associated with the observable that we will measure. The von Neumann measurement is implemented with a unitary operator that entangles the quantum system $Q$ with an ancilla $A_1$~\footnote{We focus here on orthogonal measurements, a special case of the more general POVMs (positive operator-valued measures) that use non-orthogonal states. What follows can be extended to POVMs, while at the same time Neumark's theorem guarantees that any POVM can be realized by an orthogonal measurement in an extended Hilbert space.}, 
 \begin{equation} \label{entangling-operator}
    U_{QA_1} = \sum_{x_1=1}^d P_{x_1} \!\otimes U_{x_1},
 \end{equation} 
 where $P_{x_1} = |\wt{x}_1\ra\la \wt{x}_1|$ are projectors on the state of $Q$. The operators $U_{x_1}$ transform the initial state $|0\ra$ of the ancilla to the final state $U_{x_1} |0\ra = |x_1\ra$, where $|x_1\ra$ are the orthonormal basis states of the ancilla. The unitary interaction~\eqref{entangling-operator} between the quantum system and the ancilla leads to the entangled state~\cite{CerfAdami1998}
 \begin{equation} \label{QA_stageI}
    |QA_1\ra = U_{QA_1} \, |Q\ra \, |0\ra = \sum_{x_1} \alpha^{(1)}_{x_1} \, |\wt{x}_1\ra \, |x_1\ra.
 \end{equation}
 The coefficients $\alpha^{(1)}_{x_1}$ reflect the degree of entanglement between $Q$ and $A_1$: the number of non-zero coefficients is the Schmidt number~\cite{NielsenChuang_Book} of the Schmidt decomposition.

 Tracing over~\eqref{QA_stageI}, the marginal density matrix of $A_1$ (and similarly for $Q$) is
 \begin{equation}
     \rho(A_1) = {\rm Tr}_Q \left( |QA_1\ra\la QA_1| \right) = \sum_{x_1} |\alpha^{(1)}_{x_1}|^2 \, |x_1\ra\la x_1|.
 \end{equation}
 From the symmetry of the state~\eqref{QA_stageI}, the marginal von Neumann entropy of $A_1$ is the same as $Q$, which, in turn, is equal to the Shannon entropy of the probability distribution $q^{(1)}_{x_1} = |\alpha_{x_1}^{(1)}|^2$: 
 \begin{equation} 
    S(Q) = S(A_1) = H[q^{(1)}] = - \sum_{x_1} q^{(1)}_{x_1} \, \log_d q^{(1)}_{x_1}.
 \end{equation} 
 We denote the Shannon entropy of a $d$-dimensional probability distribution $p_{x_i}$ by $H[p] = - \sum_{x_i=1}^d p_{x_i} \log_d p_{x_i}$. The von Neumann entropy of a density matrix $\rho(X)$ is defined as $S(X) = S(\rho(X)) = - {\rm Tr} \left[ \rho(X) \log_d \rho(X) \right]$, which on account of the logarithm to the base $d$, gives entropies the units ``dits". 
 
 The ancilla and quantum system are not classically correlated in~\eqref{QA_stageI} (as is required for decoherence models), but in fact are entangled. This entanglement is characterized by a negative conditional entropy~\cite{CerfAdami_PRL1997,CerfAdami1998}, $S(A_1|Q) = S(QA_1) - S(Q) = -S(A_1)$, where the joint entropy vanishes since~\eqref{QA_stageI} is pure. We illustrate the entanglement between $A_1$ and $Q$ with an entropy Venn diagram~\cite{CerfAdami1998} in Fig.~\ref{QAvenn}(a). The mutual entropy at the center of the diagram, $S(Q:A_1) = S(Q) + S(A_1) - S(QA_1)$, reflects the entropy that is shared between both systems and is twice as large as the classical upper bound~\cite{CerfAdami_PRL1997,AdamiCerf-capacity-1997,CerfAdami1998}.

 \subsection{Unprepared quantum states}
 \label{sec:unprepared-states}
 
 In the previous section, we considered measurements of a quantum system that is prepared in a pure state~\eqref{Q-prepared} with amplitudes $\alpha^{(1)}_{x_1}$. Suppose instead that we are given a quantum system about which we have no information, that is, where no previous measurement results could inform us of the state of $Q$. In this case, we write the quantum system's initial state as a maximum entropy mixed state
 \begin{equation} 
    \rho(Q) = \frac1d \sum_{x_0=1}^d |\wt{x}_0\ra\la \wt{x}_0 |\,,
 \end{equation}
 with amplitudes that now correspond to a uniform probability distribution. We call this an {\em unprepared} quantum system. We can ``purify" $\rho(Q)$ by defining a higher-dimensional pure state where $Q$ is entangled with a reference system $R$~\cite{NielsenChuang_Book},
 \begin{equation} \label{QR}
    |QR\ra= \frac{1}{\sqrt d} \sum_{x_0=1}^d |\wt{x}_0\ra |x_0\ra\,,
 \end{equation}
 such that $\rho(Q)$ is recovered by tracing~\eqref{QR} over $R$. Here and earlier, the states of $Q$ are written with a tilde, $|\wt{x}_0\ra$, to distinguish them from the states of $R$, which are denoted by $|x_0\ra$. In this section, we assume that $Q$ is an unprepared (or ``unknown'') state with maximum entropy so that it is maximally entangled with $R$, as in~\eqref{QR}. With such an assumption, we do not bias any subsequent measurements~\cite{Wootters2006}.

 To measure $Q$ with an ancilla $A_1$, we express the quantum system in the eigenbasis $|\wt{x}_1\ra$, which corresponds to the observable that ancilla $A_1$ will measure, using the unitary matrix $U^{(1)}_{x_0x_1}=\la \wt{x}_1|\wt{x}_0\ra$. 
 
 The orthonormal basis states of the ancilla, $|x_1\ra$, with $x_1 = 1 ,\ldots, d$, automatically serve as the ``interpretation basis"~\cite{Deutsch1984}. We then entangle~\cite{CerfAdami1998} $Q$ with $A_1$, which is in the initial state $|0\ra$, using the unitary entangling operation $U_{QA_1}$ in Eq.~\eqref{entangling-operator},
 \begin{equation}
  \begin{split} 
    |QRA_1 \ra & = \mathbbm{1}_{\!R} \otimes U_{QA_1} \, |QR\ra \, |0\ra \\
    & = \frac{1}{\sqrt d} \sum_{x_0x_1} \, U^{(1)}_{x_0x_1} \, |\wt{x}_1\ra \, |x_0\ra \, |x_1\ra,
  \end{split}
 \end{equation}
 where $\mathbbm{1}_{\!R}$ is the identity operation on $R$. We always write the states on the right hand side in the same order as they appear in the ket on the left hand side. We express the reference's states in terms of the $A_1$ basis by defining $|x_1\ra_{R} = \sum_{x_0} U^{(1)\intercal}_{x_1x_0} \, |x_0\ra_{R}$ with the transpose of $U^{(1)}$, so that the joint system $QRA_1$ appears as 
 \begin{equation}\label{schmidt}
    |QRA_1\ra=\frac1{\sqrt d}\sum_{x_1}|\wt{x}_1\ra |x_1\ra |x_1\ra\;. 
 \end{equation}
 Note that (\protect\ref{schmidt}) is a tripartite Schmidt decomposition of the joint density matrix {$\rho(QRA_1) = |QRA_1\ra\la QRA_1|$}, which is possible here because the entanglement operator $U_{QA_1}$ ensures the bi-Schmidt basis {$_R\la x_1|QRA_1\ra $} has Schmidt number one~\protect\cite{Pati2000}. 
 \begin{figure} % Fig.1
    \centering
    \includegraphics[width=\linewidth]{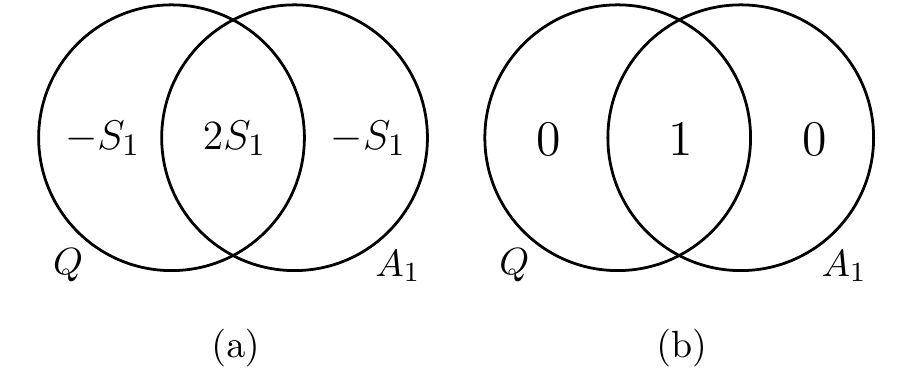} %venn_QA-v2-margin.pdf
    \caption{Entropy Venn diagrams~\cite{CerfAdami1998} for the quantum system and ancilla. (a) For prepared quantum states, $Q$ and $A_1$ are entangled according to Eq.~\eqref{QA_stageI}. (b) For unprepared quantum states, $Q$ and $A_1$ are correlated according to Eq.~\eqref{QA1-classical-correlation} when the reference $R$ has been traced out. In this figure, we use the notation $S_1 = S(A_1)$ for the marginal entropy of ancilla $A_1$.}
    \label{QAvenn}
 \end{figure}

 Tracing out the reference system from the full density matrix $\rho(QRA_1)$, we note that the ancilla is perfectly correlated with the quantum system, 
 \begin{equation} \label{QA1-classical-correlation}
     \rho(QA_1) = \frac1d \sum_{x_1} |\wt{x}_1 \, x_1 \ra\la \wt{x}_1 \, x_1 |,
 \end{equation}
 in contrast to Eq.~\eqref{QA_stageI} where $A_1$ and $Q$ are entangled. Such correlations are indicated by a vanishing conditional entropy~\cite{CerfAdami1998}, $S(A_1|Q) = S(QA_1) - S(Q) = 0$. Tracing over~\eqref{QA1-classical-correlation}, we find that each system has maximum entropy $S(Q) \!=\! S(A_1) \!=\! 1$. In Fig.~\ref{QAvenn}, we compare the entropy Venn diagrams that are constructed from the states~\eqref{QA_stageI} and~\eqref{QA1-classical-correlation}.
 
 We note in passing that $R$ can be thought of as representing all previous measurements of the quantum system that have occurred before $A_1$. We contrast measurements of unprepared quantum states~\eqref{QR} as described in this section, with measurements of \emph{prepared} quantum states (see Sec.~\ref{sec:prepared-states}), which are initially pure states~\eqref{Q-prepared} defined without a reference system $R$.

 \subsection{Composition of the quantum ancilla} 
 \label{sec:ancilla-composition}
 
 The ancilla $A_1$ may, in practice, be composed of many qudits $A^{(1)}_1\!\ldots A_1^{(n)}$, which all measure $Q$ in some basis, according to the sequence of entangling operations $U_{\!Q A_1^{(n)}} \ldots U_{\!Q A_1^{(1)}}$ between $Q$ and $A^{(i)}_1$ (see Fig.~\ref{quantum-ancilla}). In this case, Eq.~\eqref{QA_stageI}, for example, is extended to
 \begin{equation}
 \label{QA_stageI_extended}
    |Q A_1\ra = \sum_{x_1} \alpha^{(1)}_{x_1} \, |\wt{x}_1\ra  \, |x_1\ra_{\!A_1^{(1)}} \ldots |x_1\ra_{\!A_1^{(n)}}\, .
 \end{equation} 
 Tracing out the quantum system from Eq.~\eqref{QA_stageI_extended}, the joint state of the entire ancilla is $\rho(A_1) \!=\! \rho(A_1^{(1)}\!\ldots A_1^{(n)}) \!=\! \sum_{x_1} |\alpha^{(1)}_{x_1}|^2 |x_1\ldots x_1 \ra\la x_1 \ldots x_1|$. That is, each component of $A_1$ is perfectly correlated with every other component, so that $A_1$ is internally self-consistent (``all parts of $A_1$ tell the same story"). However, while $A_1$ appears classical, and could conceivably consist of a macroscopic number of components, it is potentially {\em fragile}, in the sense that its entanglement with other devices may become hidden when any part  $A_1^{(i)}$ of $A_1$ is lost (traced over). In the following, we will distinguish ``amplifiable" from non-amplifiable devices. That is, a state is amplifiable if tracing over any of its components does not affect the correlations between its subsystems.
 
 \begin{figure} % Fig.2
    \centering
    \includegraphics[width=\linewidth]{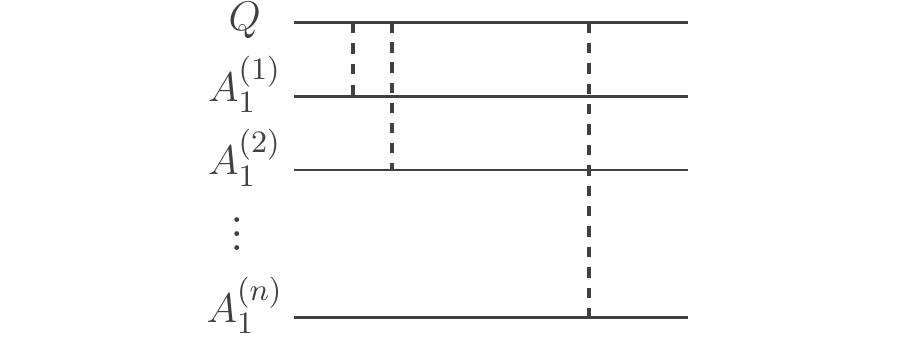} %measurement-1-margin.pdf
    \caption{Composition of the quantum ancilla. Dashed lines indicate entanglement between the quantum system, $Q$, and each of the $n$ qudits, $A_1^{(1)}\!, \ldots, A_1^{(n)}\!$, in the ancilla $A_1$.}
    \label{quantum-ancilla}
 \end{figure}

 To do this, we will consider in our discussion of Markovian quantum measurements in Sec.~\ref{sec:markov-prepared}, an additional step to the measurement process by introducing a macroscopic detector $D_1$ that measures the quantum ancilla $A_1$. In other words, $D_1$ observes the quantum observer $A_1$. This second system, which is also composed of many qudits, amplifies the measurement with $A_1$, by recording the outcome on a macroscopic device. While $A_1$ may be fragile depending on the situation, $D_1$ is robust: any part of $D_1$ could be traced over without affecting its correlations with other macroscopic measurement devices. While such a procedure (a quantum system observed by a quantum ancilla, which is observed by a classical device) may appear arbitrary, it merely represents a convenient way of splitting up the second stage of von Neumann's measurement~\cite{vonNeumann1932} to better keep track of the fate of entanglement.

 In the following sections, we formally define the concept of a quantum Markov chain that we use in this paper, in the context of consecutive measurements of a quantum system. We also further develop the formalism to describe unamplified measurements with quantum ancillae, $A_i$, which we will show are non-Markovian, and amplified measurements with macroscopic detectors, $D_i$, which are Markovian. The relationship between amplifiability and the Markov property will be the subject of Theorem~\ref{theorem-markov} in Sec.~\ref{sec:markov-chain-proof}.

 \section{Non-Markovian quantum measurements}
  \label{sec:three}
 In the previous section, we introduced the concept of non-Markovian measurements as those sequences of measurements that are not amplified by macroscopic devices, which we called $D$. In preparation for Theorem~\ref{theorem-markov} in Sec.~\ref{sec:markov-chain-proof} that establishes this correspondence, we first consider consecutive measurements with quantum ancillae of prepared and unprepared quantum states, and demonstrate the non-Markovian character of the chain of ancillae. In particular, we will use entropy Venn diagrams to study the correlations between subsystems and the distribution of entropies during consecutive measurements.

 \subsection{Consecutive measurements of a prepared quantum state}  \label{sec:prepared-states-consecutive}
 
 Building on the discussion from Sec.~\ref{sec:prepared-states} where we described a single measurement of a quantum system, we now introduce a second ancilla $A_2$ that measures $Q$. This measurement corresponds to a new basis, $|\wt{x}_2\ra$, that is rotated with respect to the old basis, $|\wt{x}_1\ra$, via the unitary transformation $U^{(2)}_{x_1x_2} = \la \wt{x}_2|\wt{x}_1\ra$. Unitarity requires that
 \begin{equation} 
 \label{eqn:unitarity-2}
 \begin{split} 
    \sum_{x_2} U^{(2)}_{x_1x_2} \, U^{(2)*}_{x'_1x_2} &= \delta_{x_1 x'_1}, \\
    \sum_{x_1} U^{(2)}_{x_1x_2} \, U^{(2)*}_{x_1x'_2} &= \delta_{x_2 x'_2}.
 \end{split} 
 \end{equation}
 After entangling $Q$ and $A_2$ with an operator analogous to~\eqref{entangling-operator}, the wavefunction~\eqref{QA_stageI} evolves to
 \begin{equation}\label{qab-prepared}
    |QA_1A_2\ra = \sum_{x_1x_2} \alpha^{(1)}_{x_1} \, U^{(2)}_{x_1x_2} \, |\wt{x}_2 \, x_1 x_2\ra,
 \end{equation}
 where the eigenstates of the second ancilla, $A_2$, are $|x_2\ra$. 
 
 Tracing out $Q$, the quantum ancillae are correlated according to the joint density matrix,
 \begin{equation}\label{AB-prepared}
    \rho(A_1 A_2) = \!\!\!\sum_{x_1 x'_1 x_2}\!\! \alpha^{(1)}_{x_1}  \alpha_{x'_1}^{(1)*} \, U^{(2)}_{x_1x_2} \, U^{(2)*}_{x'_1 x_2} \, |x_1 x_2\ra\la x'_1 x_2|,
 \end{equation} 
 while $A_1$ and $A_2$ together are entangled with the quantum system. The marginal ancilla density matrices, obtained from~\eqref{AB-prepared}, are
 \begin{equation} 
    \rho(A_i) = \sum_{x_i} q^{(i)}_{x_i} \, |x_i\ra\la x_i|\,, ~~~~ i = 1,2
 \end{equation} 
 where $q^{(1)}_{x_1} = |\alpha^{(1)}_{x_1}|^2$ is the probability distribution of ancilla $A_1$, while the probability distribution of $A_2$ is the incoherent sum $q^{(2)}_{x_2} = \sum_{x_1} |\alpha^{(1)}_{x_1}|^2 \, |U^{(2)}_{x_1x_2}|^2$. We can compare this expression to the coherent probability distribution $\sum_{x_1} |\alpha^{(1)}_{x_1} \, U^{(2)}_{x_1x_2}|^2$ for $A_2$ had the first measurement with $A_1$ never occurred. The marginal entropy of both $A_1$ and $A_2$ is the Shannon entropy $S(A_i) = H[q^{(i)}]$ of the probability distribution $q^{(i)}_{x_i}$.
 
 A third measurement of $Q$ with an ancilla $A_3$ yields
 \begin{equation}\label{qabc-prepared}
    |QA_1A_2A_3\ra = \!\!\!\sum_{x_1x_2x_3}\!\! \alpha^{(1)}_{x_1} \, U^{(2)}_{x_1x_2} \, U^{(3)}_{x_2x_3}\, |\wt{x}_3 \, x_1 x_2 x_3\ra,
 \end{equation}
 where $U^{(3)}_{x_2x_3} = \la \wt{x}_3|\wt{x}_2\ra$, and  $|x_3\ra$ are the basis states of ancilla $A_3$. The quantum system is entangled with all three ancillae in~\eqref{qabc-prepared}, as illustrated by the negative conditional entropies in Fig.~\ref{fig:qabc-prepared}. The degree of entanglement is controlled by the marginal entropy $S(A_3) = H[q^{(3)}]$ of ancilla $A_3$, for the probability distribution $q^{(3)}_{x_3} = \sum_{x_1 x_2} |\alpha^{(1)}_{x_1}|^2 \, |U^{(2)}_{x_1x_2}|^2 \, |U^{(3)}_{x_2x_3}|^2$. This procedure can be repeated for an arbitrary number of consecutive measurements and can be used to succinctly describe the quantum Zeno and anti-Zeno effects (see Sec.~\ref{sec:zeno}).
 
 \begin{figure} % Fig. 3
  \centering
  \includegraphics[width=\linewidth]{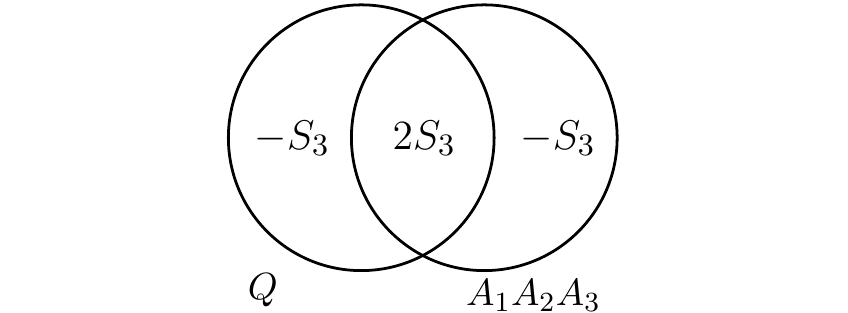} %venn_QABC-nonmarkov-margin.pdf
  \caption{Entropy Venn diagram for the state~\eqref{qabc-prepared}. The presence of negative conditional entropies reveals that the quantum system, $Q$, is entangled with all three ancillae, $A_1A_2A_3$. In this figure, we use the notation $S_3 = S(A_3)$, which is the marginal entropy of the last ancilla $A_3$ in the measurement sequence. To generalize this diagram from three to $n$ consecutive measurements of a prepared quantum system, $S_3$ is replaced by $S_n$, the entropy of the last ancilla, $A_n$, in the chain.}
  \label{fig:qabc-prepared}
\end{figure}

 \subsection{Consecutive measurements of an unprepared quantum state}
 
 Sequential measurements of an unprepared quantum system yield entropy distributions between the quantum system and ancillae that are different from those created by measurements of prepared quantum systems (see Sec.~\ref{sec:prepared-states-consecutive}). In this section, we consider a sequence of measurements of an unprepared quantum system that is initially entangled with a reference system as in~\eqref{QR}. Adding to the calculations in Sec.~\ref{sec:unprepared-states}, we measure $Q$ again in a rotated basis $U^{(2)}_{x_1x_2} = \la \wt{x}_2|\wt{x}_1\ra$, by entangling it with an ancilla $A_2$. Then, with $|x_2\ra$ the basis states of ancilla $A_2$, the wavefunction~\eqref{schmidt} becomes
 \begin{equation} \label{QRAB}
    |Q R A_1 A_2 \ra = \frac{1}{\sqrt{ d}} \sum_{x_1 x_2} U^{(2)}_{x_1x_2} \, |\wt{x}_2 \, x_1 x_1x_2\ra \,.
 \end{equation}  
 It is straightforward to show that the marginal ancilla density matrices are maximally mixed, $\rho(A_1) \!=\! \rho(A_2) \!=\! 1/d \, \mathbbm{1}$, where $\mathbbm{1}$ is the identity matrix of dimension $d$. It follows that $A_1$ and $A_2$ have maximum entropy $S(A_1) = S(A_2) =1$~\footnote{Recall that all logarithms are taken to the base $d$, giving entropies the units {\em dits}. If $d=2$, the units are {\em bits}.} . The joint state of $A_1$ and $A_2$ is diagonal in the ancilla product basis,
 \begin{equation}\label{rhoab}
    \rho(A_1A_2)=\frac1d\sum_{x_1}|x_1\ra\la x_1| \otimes \!\sum_{x_2} |U^{(2)}_{x_1x_2}|^2 |x_2\ra\la x_2| \, , 
 \end{equation}
 in contrast to Eq.~\eqref{AB-prepared}. Still, the quantum ancillae $A_1$ and $A_2$ are correlated. Equations~\eqref{AB-prepared} and~\eqref{rhoab} immediately imply that if the quantum system is measured repeatedly in the same basis ($U^{(2)}_{x_1x_2} = \delta_{x_1x_2}$) by independent devices, all of those devices will be perfectly correlated and will reflect the same outcome~\cite{CerfAdami1996,CerfAdami1998}.
 
 Let us entangle a third ancilla, $A_3$, with the quantum system such that $U^{(3)}_{x_2x_3} = \la \wt{x}_3|\wt{x}_2\ra$. We find that~\eqref{QRAB} evolves to
 \begin{equation}\label{qabc}
    |QRA_1A_2A_3\ra =\frac1{\sqrt d}  \sum_{x_1x_2x_3} \! U^{(2)}_{x_1x_2} \, U^{(3)}_{x_2x_3} \, |\wt{x}_3 \, x_1 x_1 x_2 x_3\ra.
 \end{equation}
 The entropic relationships between the variables $Q$, $R$, and $A_1A_2A_3$ are shown in Fig.~\ref{fig:qrabc}. The zero ternary mutual entropy, $S(Q:R:A_1A_2A_3) = 0$, indicates that the entropy $S(A_1A_3) = S_{13}$ that is shared by $R$ and $A_1A_2A_3$ is not shared with the quantum system. Tracing out the reference state, we find that the quantum system is entangled with all three ancillae. However, this entanglement is now shared with the reference system, which yields a Venn diagram that is different from Fig.~\ref{fig:qabc-prepared}.
 
 \begin{figure} % Fig. 4
    \centering
    \includegraphics[width=\linewidth]{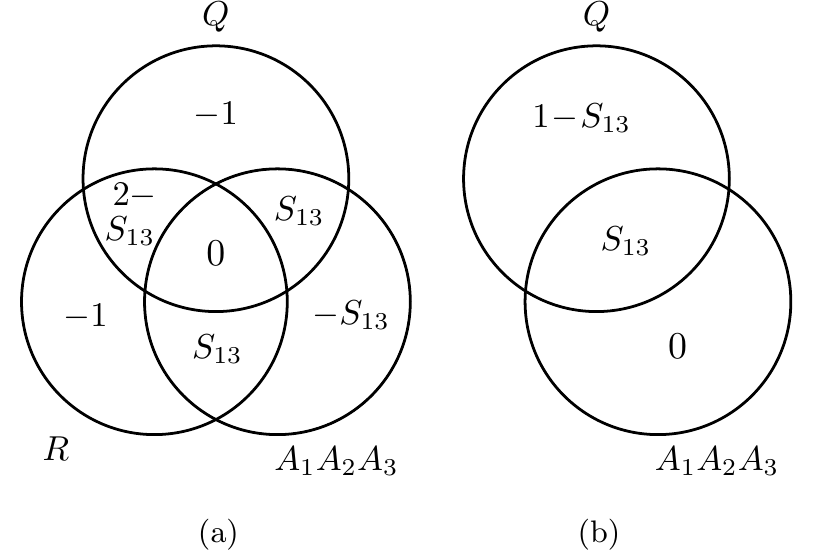} %venn_QRABC-nonmarkov-margin.pdf
    \caption{Entropy Venn diagrams for the pure state~\eqref{qabc}, where $S_{13} = S(A_1A_3)$ is the joint entropy of Eq.~\eqref{rhoac}. (a) The entropy $S_{13}$ that is shared by the reference state, $R$, and the chain of ancillae, $A_1A_2A_3$, is not shared with the quantum system, $Q$, since the ternary mutual entropy vanishes: $S(Q:R:A_1A_2A_3) = 0$. (b) Tracing out the reference leaves $Q$ entangled with all ancillae. To generalize these diagrams from three to $n$ consecutive measurements of an unprepared quantum system, $S_{13}$ is replaced by $S_{1n}$, the joint entropy of the first and last ancillae, $A_1$ and $A_n$, in the chain.}
    \label{fig:qrabc}
 \end{figure}
 
 Consecutive measurements provide a unique opportunity to extract information about the state of the quantum system from the correlations created between the ancillae, as we do not directly observe either the quantum system or the reference. Tracing out $Q$ and $R$ from the full density matrix associated with Eq.~\eqref{qabc} yields the joint state of the three ancillae,
 \begin{equation}\label{rhoABC}
 \begin{split}
    \rho(A_1A_2A_3) \! = \frac1d  \sum_{x_1} |x_1\ra\la x_1|  & \otimes  \!\! \sum_{x_2x'_2}  U^{(2)}_{x_1x_2} \, U^{(2)*}_{x_1x'_2} \, |x_2\ra\la x'_2| \\
    & \otimes \! \sum_{x_3}  U^{(3)}_{x_2 x_3} \, U^{(3)*}_{x'_2 x_3}  \, |x_3\ra\la x_3|.
 \end{split}
 \end{equation}
Unlike the pairwise state $\rho(A_1A_2)$ in Eq.~\eqref{rhoab}, the state of all three ancillae is not an incoherent mixture. Performing a third measurement has, in a sense, revived the coherence of the $A_2$ subsystem.  
  
 An apparent collapse has taken place after the second consecutive measurement in Eq.~\eqref{rhoab} as the corresponding density matrix has no off-diagonal terms. However, the third measurement seemingly {\em undoes} this projection, as can be seen from the appearance of off-diagonal terms in Eq.~\eqref{rhoABC}. This ``reversal" is different from protocols that can ``un-collapse" weak measurements~\cite{korotkov2006,Jordan2010}, because it is clear that the wavefunction~\eqref{qabc} underlying the density matrix~\eqref{rhoABC} was never projected after all. The presence of the cross terms in Eq.~\eqref{rhoABC} has fundamental consequences for our understanding of the measurement process, and may open up avenues for developing new quantum protocols. In particular, the cross terms in Eq.~\eqref{rhoABC} enable the implementation of disentangling protocols~\cite{GlickAdami2017}.
 
As mentioned in Sec.~\ref{sec:ancilla-composition}, the ancilla $A_i$ may be composed of a large number of qudits. To account for a possibly macroscopic ancilla, we suppose that $n$ qudits $A_i^{(1)}\cdots A_i^{(n)}$, which comprise the $i$th ancilla $A_i$, measure the quantum system in the same given basis. In this case, the joint density matrix~\eqref{rhoABC} is extended to
 \begin{widetext}
 \begin{equation} \label{rhoABC2}
    \rho(A_1A_2A_3) = \frac1d  \sum_{x_1} |x_1 \ldots x_1 \ra\la x_1 \ldots x_1| \otimes  \! \sum_{x_2x'_2}  U^{(2)}_{x_1x_2} \, U^{(2)*}_{x_1x'_2} \, |x_2 \ldots x_2 \ra\la x'_2 \ldots x'_2| \otimes  \sum_{x_3}  U^{(3)}_{x_2 x_3} \, U^{(3)*}_{x'_2 x_3} \, |x_3 \ldots x_3 \ra\la x_3 \ldots x_3|.
 \end{equation} 
 \end{widetext}
 In principle, accounting for macroscopic ancillae does not destroy the coherence of the joint state~\eqref{rhoABC2}, which is concentrated in the $A_2$ subsystem. The coherence is protected as long as no qudits in the intermediate ancilla $A_2$ are `lost', implying a trace over their states, which removes all off-diagonal terms. In practical implementations, it may be effectively impossible to prevent decoherence when the number of qudits is sufficiently large. On the other hand, the pairwise density matrices $\rho(A_1A_2)$, $\rho(A_2A_3)$, and $\rho(A_1A_3)$ are unaffected by a loss of qudits as they are already diagonal. In addition, it can be easily shown that the coherence in Eqs.~\eqref{rhoABC} and~\eqref{rhoABC2} is fully destroyed if just the $A_2$ measurement is amplified by a detector $D_2$ (as we will see in Sec.~\ref{sec:markov-unprepared}). That is, amplification of the first and last ancillae has no effect on the coherence of~\eqref{rhoABC} and~\eqref{rhoABC2}.

 From the joint ancilla density matrix~\eqref{rhoABC}, we now derive several properties of the chain of quantum ancillae and summarize them using an entropy Venn diagram between $A_1$, $A_2$, and $A_3$. First, we construct all three pairwise ancilla density matrices and compute their entropies. Tracing out $A_3$ from the joint density matrix~\eqref{rhoABC} recovers $\rho(A_1A_2)$ in Eq.~(\ref{rhoab}), as it should because the interaction between $Q$ and $A_3$ does not affect the past interactions of $Q$ with $A_1$ and $A_2$. Tracing over $A_2$ in Eq.~\eqref{rhoABC} gives
 \begin{equation} \label{rhoac}
    \rho(A_1A_3) = \frac1d  \sum_{x_1}  |x_1\ra\la x_1|\otimes \!\sum_{x_2x_3} |U^{(2)}_{\!x_2x_3}|^2 \,  |U^{(3)}_{\!x_2x_3}|^2 \, |x_3\ra\la x_3|,
 \end{equation} 
 while tracing over $A_1$ yields
 \begin{equation} 
 \label{rhobc}
    \rho(A_2A_3)=\frac1d \sum_{x_2} |x_2\ra\la x_2|\otimes \sum_{x_3} |U^{(3)}_{x_2x_3}|^2 \, |x_3\ra\la x_3|\;. 
 \end{equation}
 All three pairwise density matrices are diagonal in the ancilla product basis (see Theorem~\ref{theorem-cq-unprepared} in Sec.~\ref{sec:coherence} for a general proof). We take ``diagonal in the ancilla product basis" to be synonymous with ``classical". From Eqs.~\eqref{rhoab},~\eqref{rhoac}, and~\eqref{rhobc}, we can calculate the entropy of each pair of ancillae and of the joint state of all three ancillae from Eq.~\eqref{rhoABC}. The pairwise entropies are
 \begin{align}
    S(A_1A_2)  & =  1  -  \frac{1}{d} \!\sum_{x_1x_2} |U^{(2)}_{x_1x_2}|^2 \log_d |U^{(2)}_{x_1x_2}|^2,  \label{SAB}\\
    S(A_2A_3)  & =  1  -  \frac{1}{d} \! \sum_{x_2x_3} |U^{(3)}_{x_2x_3}|^2 \log_d |U^{(3)}_{x_2x_3}|^2, \label{SBC}\\
    S(A_1A_3)  & =  1  -  \frac{1}{d} \! \sum_{x_1x_3} |\beta^{(13)}_{x_1x_3}|^2 \log_d |\beta^{(13)}_{x_1x_3}|^2. \label{SAC}
 \end{align}
where $|\beta^{(13)}_{x_1x_3}|^2 = \sum_{x_2} |U^{(2)}_{x_1x_2}|^2  |U^{(3)}_{x_2x_3}|^2$. Furthermore, it is straightforward to show that $S(A_1A_2A_3)$, the entropy of $\rho(A_1A_2A_3)$ in Eq.~\eqref{rhoABC}, is equal to $S(A_1A_3)$. This equality holds for any set of three consecutive measurements in an arbitrarily-long measurement chain as we will later prove in Theorem~\ref{theorem-cq-unprepared} of Sec.~\ref{sec:coherence}. With these joint entropies, we construct the entropy Venn diagram for the three ancillae that consecutively measured an unprepared quantum system, as shown in Fig.~\ref{fig:venn-ABC-nonmarkov}.
 
  \begin{figure} % Fig. 5
    \centering
    \includegraphics[width=\linewidth]{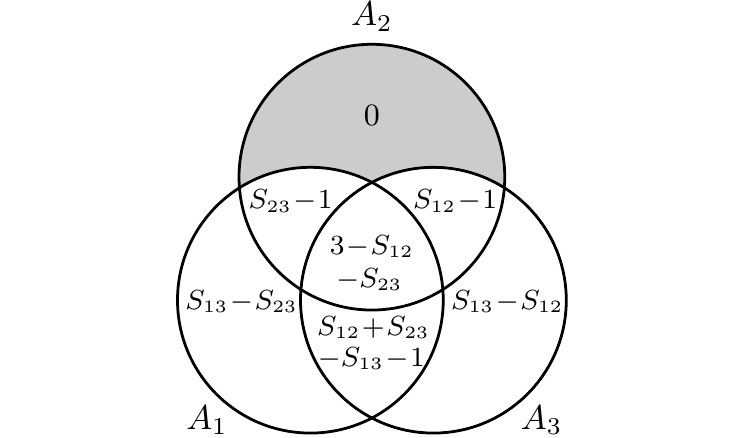} %venn_ABC-nonmarkov-margin.pdf
    \caption{Entropy Venn diagram for three quantum ancillae that measured an unprepared quantum system. In this figure, we use the notation $S(A_iA_j) = S_{ij}$ for the pairwise entropy of any two ancillae $A_i$ and $A_j$.}
    \label{fig:venn-ABC-nonmarkov} 
 \end{figure}
 
We apply the formalism presented thus far to the specific case of qubits (Hilbert space dimension $d \!=\! 2$). Measurements with ancilla $A_2$ at an angle $\theta_2$ relative to the previous measurement with $A_1$, and with ancilla $A_3$ at an angle $\theta_3$ relative to $A_2$, can each, without loss of generality, be implemented with a rotation matrix of the form  
 \begin{equation}
 \label{rotation}
        U^{(i)} = \left( \! \begin{array}{cc}\cos(\theta_i)& -\sin(\theta_i)\\
                           \sin(\theta_i)&  \;\;\; \cos(\theta_i)
              \end{array} \! \right) \: .
 \end{equation}
 For measurements at $\theta_2 \!=\! \theta_3 \!=\! \pi/4$, for example, we have $|U^{(2)}_{x_1x_2}|^2 \!=\! |U^{(3)}_{x_2x_3}|^2 \!=\! 1/2$, and we expect each ancilla to be maximally entropic: $S(A_1) \! = \! S(A_2) \! = \! S(A_3) \! = \! 1$ bit. The joint entropy of each pair of ancillae is two bits, as can be read off of Eqs.~(\ref{SAB}-\ref{SAC}). Because of the non-diagonal nature of $\rho(A_1A_2A_3)$ in Eq.~\eqref{rhoABC}, the joint density matrix of the three ancillae (using $\sigma_z$, the third Pauli matrix, and $\mathbbm{1}$, the $2\times2$ identity matrix),
 \begin{equation}\label{three} 
    \rho(A_1A_2A_3) = \frac{1}{8} 
    \begin{pmatrix}
    \mathbbm{1}  & -\sigma_z  & 0           & 0          \\
    -\sigma_z   & \mathbbm{1} & 0           & 0          \\ 
    0           & 0          & \mathbbm{1}  & \sigma_z   \\
    0           & 0          & \sigma_z    & \mathbbm{1} 
    \end{pmatrix} , 
 \end{equation}
 has entropy $S(A_1A_2A_3)=2$ bits, as can be checked by finding the eigenvalues of (\ref{three}). Figure~\ref{fig:venn-abc-qubit-nonmarkov} summarizes the entropic relationships for unamplified consecutive qubit measurements at $\theta_2=\theta_3=\pi/4$.

 It is instructive to note that the Venn diagram in Fig.~\ref{fig:venn-abc-qubit-nonmarkov} is the same as the one obtained for a one-time binary cryptographic pad where two classical binary variables (the source and the key) are combined to a third (the message) via a controlled-\textsc{not} operation~\cite{Schneidman2003} (the density matrices underlying the Venn diagrams are very different, however). The Venn diagram implies that the state of any one of the three quantum ancillae can be predicted from knowing the {\em joint} state of the two others. However, the prediction of $A_3$, for example, cannot be achieved using expectation values from $A_2$'s and $A_1$'s states separately, as the diagonal of Eqs.~\eqref{rhoABC} and~\eqref{three} corresponds to a uniform probability distribution. Thus, quantum coherence can be seen to encrypt classical information about past states. 
 \begin{figure} % Fig. 6
    \centering
    \includegraphics[width=\linewidth]{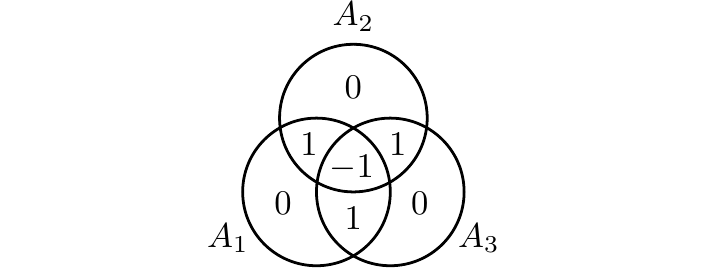} %venn_ABC-qubit-nonmarkov-margin.pdf
    \caption{Entropy Venn diagram for three qubit ancillae that measured an unprepared quantum system. Ancilla $A_2$ measured $Q$ at an angle $\theta_2=\pi/4$ relative to the basis of $A_1$, and $A_3$ measured $Q$ at $\theta_3=\pi/4$ relative to the basis of $A_2$.}
   \label{fig:venn-abc-qubit-nonmarkov}
 \end{figure}

 \subsection{Coherence of the chain of unamplified measurements}  
 \label{sec:coherence}
 
 So far we have seen that the joint ancilla density matrices describing unamplified measurements generally contain a non-vanishing degree of coherence. This suggests that coherence is not lost in the measurement sequence, but is actually contained in specific ancilla subsystems. In this section, we extend our unitary description of consecutive unamplified measurements of a quantum system to an arbitrarily-long chain of ancillae, and derive several properties of the measurement chain.
 
 Many of the joint ancilla density matrices that we have encountered in describing consecutive quantum measurements are so-called ``classical-quantum states". Such states have a block-diagonal structure of the form $\rho = \sum_i p_i \, \rho_i \otimes |i\ra\la i|$, where the density matrix $\rho_i$ appears with probability $p_i$. However, the ancilla states that we derive here have the additional property that the density matrices $\rho_i$ are always pure quantum superpositions. 
 
 For measurements of a prepared quantum system, classical-quantum states occur in the joint density matrices of two or more consecutive ancillae. For instance, recall the state $\rho(A_1A_2)$ from Eq.~\eqref{AB-prepared} that resulted from two measurements of a prepared quantum system. We can diagonalize this state with the set of non-orthogonal states $\alpha^{(2)}_{x_2} \, |\psi_{x_2}\ra = \sum_{x_1} \alpha^{(1)}_{x_1} \, U^{(2)}_{x_1x_2} \, |x_1\ra$ for subsystem $A_1$, so that~\eqref{AB-prepared} appears as
 \begin{equation} \label{rhoab-diagonal}
    \rho(A_1A_2) = \sum_{x_2} q^{(2)}_{x_2} \, |\psi_{x_2}\ra\la\psi_{x_2}| \otimes |x_2\ra\la x_2|,
 \end{equation} 
 where the normalization is equal to the probability distribution for the second ancilla $A_2$,
 \begin{equation} \label{diagonalize-A1A2}
    q^{(2)}_{x_2} = |\alpha^{(2)}_{x_2}|^2 = \sum_{x_1} |\alpha^{(1)}_{x_1}|^2 \, |U^{(2)}_{x_1x_2}|^2 .
 \end{equation} 
 
 On the other hand, classical-quantum states occur for measurements of unprepared quantum systems when there are at least {\em three} consecutive measurements, as the first measurement in that sequence can be viewed as the state preparation. For example, Eq.~\eqref{rhoABC} can be diagonalized with the set of non-orthogonal states $\beta^{(13)}_{x_1x_3} \, |\phi_{x_1x_3}\ra = \sum_{x_2} U^{(2)}_{x_1x_2} \, U^{(3)}_{x_2x_3} \, |x_2\ra$ for the $A_2$ subsystem, so that 
 \begin{equation}\label{rhoABC-diagonal}
    \rho(A_1A_2A_3) \!= \frac1d \sum_{x_1x_3} p^{(13)}_{x_1x_3} \,  |x_1x_3\ra\la x_1x_3| \otimes |\phi_{x_1x_3}\ra\la\phi_{x_1x_3}| ,
 \end{equation} 
 where the normalization is
 \begin{equation}
    p^{(13)}_{x_1x_3} = |\beta^{(13)}_{x_1x_3}|^2 = \sum_{x_2} |U^{(2)}_{x_1x_2}|^2 \, |U^{(3)}_{x_2x_3}|^2.
 \end{equation} 
 
 Evidently, from~\eqref{rhoab-diagonal} and~\eqref{rhoABC-diagonal}, each density matrix $\rho_i$ in the general state $\rho = \sum_i p_i \, \rho_i \otimes |i\ra\la i|$ corresponds to a pure state in our ancilla density matrices. This leads to an interesting observation that the entropy of a chain of ancillae is contained in either just the last device or in both the first and last devices together. In the first example above for $\rho(A_1A_2)$, it is straightforward to show using Eq.~\eqref{rhoab-diagonal} that $S(A_1A_2) = S(A_2)$. That is, the entropy of the sequence $A_1A_2$ is found at the end of the chain, $A_2$. From the definition of conditional entropy~\cite{CerfAdami_PRL1997}, it follows that the entropy of $A_1$ vanishes (it is in the pure state $|\psi_{x_2}\ra$), given the state of $A_2$: 
 \begin{equation}
     S(A_1|A_2) = S(A_1A_2) - S(A_2) = 0.
 \end{equation}

 In the second example above for $\rho(A_1A_2A_3)$, we find from Eq.~\eqref{rhoABC-diagonal} that $S(A_1A_3) = S(A_1A_2A_3)$. In other words, the entropy of the chain resides in the boundary, $A_1$ and $A_3$. It follows that, given the joint state of $A_1$ and $A_3$, $A_2$'s state has zero entropy (see the grey region in Fig.~\ref{fig:venn-ABC-nonmarkov}) and is fully determined (it is in the pure state $|\phi_{x_1x_3}\ra$):
 \begin{equation}
    S(A_2|A_1A_3) = S(A_1A_2A_3) - S(A_1A_3) = 0\:.
 \end{equation}

 In the following Theorems~\ref{theorem-cq-prepared} and~\ref{theorem-cq-unprepared}, we extend these results to an arbitrarily-long chain of quantum ancillae. These findings are important as they show that unamplified measurement chains retain a finite amount of coherence. Specifically, for measurements on unprepared quantum states, the coherence is contained in all ancillae up to the last, while for unprepared quantum states it is contained in all ancillae except for the boundary.
 
 To begin, we define (ancilla) random variables $A_i$ that take on states $x_i$ with probabilities $q^{(i)}_{x_i}$. Each ancilla has $d$ orthogonal states and the set of outcomes for the $i$th ancilla is labeled by the index $x_i$, where $x_i = 1, \ldots , d$.

 %%%%%%%%%%%%%%%%%%%%%%%%%%%%
 % Theorem - cq - prepared  %
 %%%%%%%%%%%%%%%%%%%%%%%%%%%%
  
 \begin{thm}\label{theorem-cq-prepared}
  The density matrix describing $j+1$ ancillae that consecutively measured a prepared quantum system is a classical-quantum state such that its joint entropy is contained only in the last device in the measurement chain. That is,
  \begin{equation}
      S(A_1 \ldots A_{j+1}) = S(A_{j+1}).
  \end{equation}
 \end{thm}
 
 \begin{proof}
 Generalizing the result~\eqref{qabc-prepared}, the wavefunction $|\Psi\ra = |Q A_1 \ldots A_{j+1}\ra$ for $j+1$ consecutive measurements of a prepared quantum state is
 \begin{equation}\label{eqn:PreparedFullWavefunction}
     |\Psi\ra = \!\!\!\!\!\! \sum_{\substack{ x_1 \ldots x_{j+1} }} \!\!\!\! \alpha^{(1)}_{x_1} \,\, U^{(2)}_{x_1 x_2}  \ldots U^{(j+1)}_{x_j x_{j+1}} \, |\wt{x}_{j+1} \, x_1 x_2 \ldots x_j x_{j+1}\ra .
 \end{equation}
 The first ket $|\wt{x}_{j+1}\ra$ in the joint state on the right hand side of~\eqref{eqn:PreparedFullWavefunction} describes the quantum system, which is written in the basis of the last ancilla. Each $A_i$ measures the quantum system in a basis that is rotated relative to the basis of the previous $A_{i-1}$, such that $U^{(i)}_{x_{i-1},x_i} = \la \wt{x}_i | \wt{x}_{i-1} \ra$. The unitarity of $U^{(i)}$ requires that
 \begin{equation}  
 \label{eqn:unitarity}
 \begin{split}
    \sum_{x_{i-1}} U^{(i)}_{x_{i-1}x_i} U^{(i)*}_{x_{i-1}x'_i} &= \delta_{x_ix_i'} \, ,\\
    \sum_{x_i} U^{(i)}_{x_{i-1}x_i} U^{(i)*}_{x'_{i-1}x_i} &= \delta_{x_{i-1}x'_{i-1}} \, .
 \end{split} 
 \end{equation}
 
 Recasting expression~\eqref{eqn:PreparedFullWavefunction} in terms of the following set of non-orthogonal states,
 \begin{equation}
 \label{diagonalize-prepared-general}
      \alpha^{(j+1)}_{x_{j+1}} \,\, |\psi_{x_{j+1}}\ra = \!\!\! \sum_{x_1 \cdots x_j} \!\!\! \alpha^{(1)}_{x_{1}} \,\, U^{(2)}_{x_1x_2} \cdots \, U^{(j+1)}_{x_jx_{j+1}} \, |x_1 \cdots x_j\ra,
 \end{equation}
 yields
 \begin{equation}\label{eqn:FullWavefunction2}
    |\Psi\ra  = \sum_{x_{j+1}} \alpha^{(j+1)}_{x_{j+1}} ~ |\wt{x}_{j+1}~\psi_{x_{j+1}}~x_{j+1}\ra.
 \end{equation}
 This is not a true tripartite Schmidt decomposition~\cite{Pati2000} as the states $|\psi_{x_{j+1}}\ra$ are not orthogonal: the partial inner product $\la\psi_{x_{j+1}}|\Psi\ra$ does not give a state with a Schmidt number of one. Although the states $|\psi_{x_{j+1}}\ra$ are not orthogonal, they are normalized according to
 \begin{equation}
 \label{diagonalize-prepared-general-normalization}
      q^{(j+1)}_{x_{j+1}} = |\alpha^{(j+1)}_{x_{j+1}}|^2 = \!\!\! \sum_{x_1 \cdots x_j} \! |\alpha^{(1)}_{x_{1}}|^2 \,\, |U^{(2)}_{x_1x_2}|^2 \cdots \, |U^{(j+1)}_{x_jx_{j+1}}|^2, 
 \end{equation}
 which is the probability distribution of ancilla $A_{j+1}$. 

 Tracing out the quantum system from the density matrix $|\Psi\ra\la\Psi|$ formed from~\eqref{eqn:FullWavefunction2}, the state of all $j+1$ ancillae can be written as
 \begin{equation}\label{eqn:cq-prepared}
      \rho(A_1 \ldots A_{j+1}) = \!\! \sum_{x_{j+1}} \! q^{(j+1)}_{x_{j+1}} ~ |\psi_{x_{j+1}} \ra\la \psi_{x_{j+1}} | \otimes |x_{j+1} \ra\la x_{j+1} |.
 \end{equation}
 This state is non-diagonal in the ancilla product basis $|x_1 \cdots x_{j+1}\ra$, but is diagonalized by~\eqref{diagonalize-prepared-general}. The density matrix~\eqref{eqn:cq-prepared} is a classical-quantum state where the first $j$ ancillae are in the pure state $|\psi_{x_{j+1}} \ra$. 
 
 The appearance of classical-quantum states in the sequence of measurements leads to the interesting (and perhaps surprising) observation that the joint entropy of all ancillae in Eq.~\eqref{eqn:cq-prepared} resides only in the last device in the measurement chain. Since the joint state $|\psi_{x_{j+1}} \ra \otimes |x_{j+1}\ra$ is orthonormal, it is easy to see that the entropy of~\eqref{eqn:cq-prepared} is equal to the Shannon entropy of the probability distribution $q^{(j+1)}_{x_{j+1}}$. This is equivalent to the entropy of the last ancilla, so that
 \begin{equation}
    S(A_1 \ldots A_{j+1}) = S(A_{j+1}).
 \end{equation}
 \end{proof}
 Note that this implies that there is an upper bound to the joint entropy: ${\rm max}[S(A_1 \ldots A_{j+1})] = {\rm max}[S_{j+1}] = 1$. 
 
 \begin{figure} % Fig. 7
     \centering
     \includegraphics[width=\linewidth]{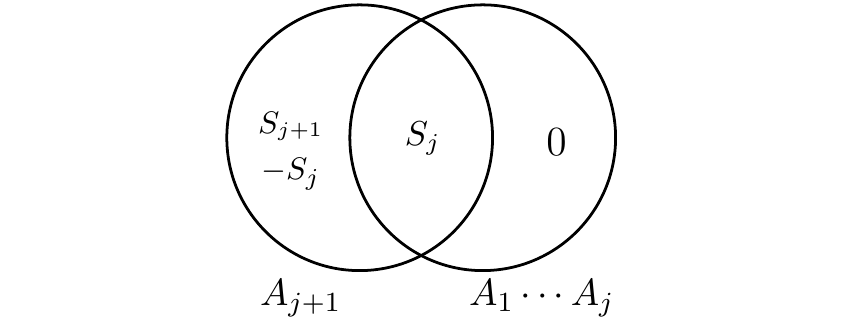} %venn_AiAj-prepared-nonmarkov-margin.pdf
     \caption{Entropy Venn diagram for the unamplified measurement sequence with ancillae $A_1, A_2, \ldots, A_j, A_{j+1}$. According to~\eqref{eqn:cq-prepared}, the joint entropy of the first $j$ ancillae vanishes when given $A_{j+1}$, since the entropy resides only at the end of the chain. In this figure, we use the notation $S_{k} = S(A_k)$ for the marginal entropy of the $k$th ancilla.}
     \label{fig:chain-prepared-arbitrary}
 \end{figure}  
 From this property, it immediately follows that the entropy of the first $j$ ancillae, conditional on the state of the last ancilla, vanishes,
 \begin{equation}
      S(A_1 \ldots A_j | A_{j+1}) = S(A_1 \ldots A_{j+1}) - S(A_{j+1}) = 0.
 \end{equation}
 Therefore, if the state of the end of the measurement chain is known, then all preceding ancillae exist in a pure quantum superposition: The state of $A_1 \ldots A_j$ is fully determined (a zero entropy state), given $A_{j+1}$. This implies that the entropy of all ancillae in an arbitrarily-long sequence of measurements resides only at the end of the chain. The entropy Venn diagram for these two subsystems is shown in Fig.~\ref{fig:chain-prepared-arbitrary}.

 %%%%%%%%%%%%%%%%%%%%%%%%%%%%%%
 % Theorem - cq - unprepared  %
 %%%%%%%%%%%%%%%%%%%%%%%%%%%%%%
 
 \begin{thm}\label{theorem-cq-unprepared}
  For $j+1$ consecutive measurements of an unprepared quantum system, where the reference is traced out, the density matrix for three or more consecutive ancillae is a classical-quantum state such that its joint entropy is contained only in the first and last device of the measurement chain. That is, 
  \begin{equation}
      S(A_{i-1} A_i \ldots A_j A_{j+1}) = S(A_{i-1} A_{j+1}).
  \end{equation}
 \end{thm}
 
 \begin{proof}
 Generalizing the result~\eqref{qabc}, the wavefunction $|\Psi'\ra = |Q R A_1\ldots A_{j+1}\ra$ of $j+1$ ancillae that consecutively measured an unprepared quantum state is
 \begin{equation}\label{eqn:UnpreparedFullWavefunction}
     |\Psi'\ra = \frac{1}{\sqrt{d}} \!\!\!\! \sum_{\substack{ ~~x_1 \ldots x_{j+1} }} \!\!\!\!\! U^{(2)}_{x_1 x_2}  \!\ldots U^{(j+1)}_{x_j x_{j+1}} |\wt{x}_{j+1} \, x_1 \, x_1 x_2 \ldots x_{j+1}\ra .
 \end{equation}
 Of the full set of consecutive measurements, consider the subset $A_{i-1}, A_i, \ldots,A_j,A_{j+1}$, where $1 < i< j$. Tracing out $Q$, the reference, and all other ancilla states from the full density matrix $|\Psi'\ra\la\Psi'|$, and using the unitarity of each $U^{(i)}$ as stated in Eq.~\eqref{eqn:unitarity}, the density matrix for this subset can be written as
 \begin{widetext}
 \begin{equation}
 \label{eqn:cq-arbitrary}
    \rho(A_{i-1} \ldots A_{j+1}) =  \frac{1}{d} \, \sum_{\substack{ x_{i-1} \\ x_{j+1}}} \, p^{(i-1,j+1)}_{x_{i-1}x_{j+1}} \,~ |x_{i-1}\ra\la x_{i-1} | \otimes |\phi_{x_{i-1}x_{j+1}} \ra\la \phi_{x_{i-1}x_{j+1}} | \otimes |x_{j+1} \ra\la x_{j+1}| \,. 
 \end{equation}
 \end{widetext}
 This is a classical-quantum state with the intermediate ancillae $A_i,\ldots,A_j$ in the pure state $|\phi_{x_{i-1}x_{j+1}} \ra$. In the ancilla product basis $|x_{i-1} x_i \ldots x_j x_{j+1}\ra$, this matrix is block-diagonal due to the non-diagonality of the subsystem $A_i, \ldots, A_j$. However, it is diagonalized by the non-orthogonal states
 \begin{equation} 
 \label{diagonalize-unprepared-general}
        \beta^{(i-1,j+1)}_{x_{i-1} x_{j+1}} \, |\phi_{x_{i-1} x_{j+1}}\ra  = \!\!\!\sum_{\substack{ x_i \cdots x_j }} \!\!  U^{(i)}_{x_{i-1} x_i} \ldots U^{(j+1)}_{x_j x_{j+1}} \, |x_i \ldots x_j \ra ,
 \end{equation}
 which are normalized according to
 \begin{equation}
 \label{diagonalize-unprepared-general-normalization}
   p^{(i-1,j+1)}_{x_{i-1} x_{j+1}} = |\beta^{(i-1,j+1)}_{x_{i-1} x_{j+1}}|^2  = \!\!\! \sum_{\substack{ x_i \cdots x_j }} \! | U^{(i)}_{x_{i-1} x_i }|^2 \ldots |U^{(j+1)}_{x_j x_{j+1} }|^2.
 \end{equation} 
 These normalization coefficients obey the sum rule 
 \begin{equation}
      \sum_{\substack{ x_{i-1} }} p^{(i-1,j+1)}_{x_{i-1} x_{j+1}} = \sum_{\substack{ x_{j+1} }} p^{(i-1,j+1)}_{x_{i-1} x_{j+1}}  = 1\,.
 \end{equation}

 The density matrix for any two ancillae is already diagonal in the ancilla product basis (it is classical). For example, the joint state of $A_{i-1}$ and $A_{j+1}$ is
 \begin{equation}\label{eqn:cq-arbitrary-pairwise}
      \rho(A_{i-1} A_{j+1}) = \frac1d \sum_{\substack{ x_{i-1} \\ x_{j+1} }}  p^{(i-1,j+1)}_{x_{i-1} x_{j+1}} ~ |x_{i-1} x_{j+1}\ra\la x_{i-1} x_{j+1}|,
 \end{equation}
 so that its entropy reduces to the Shannon entropy $H[p^{(i-1,j+1)}/d]$ of the distribution $p^{(i-1,j+1)}_{x_{i-1}x_{j+1}}/d$. However, the density matrix for three or more consecutive ancillae corresponds to a classical-quantum state~\eqref{eqn:cq-arbitrary}. This state has non-zero coherence that is contained in the subsystem of the intermediate ancillae, which are in the (non-orthogonal) pure state $|\phi_{x_{i-1}x_{j+1}}\ra$. Since the joint state $|x_{i-1}\ra\otimes |\phi_{x_{i-1}x_{j+1}} \ra \otimes |x_{j+1}\ra$ is still orthonormal, it is straightforward to show that the entropy of~\eqref{eqn:cq-arbitrary} is equal to the (Shannon) entropy of~\eqref{eqn:cq-arbitrary-pairwise}, despite the fact that the underlying state~\eqref{eqn:cq-arbitrary} is non-classical:
 \begin{equation}
    S(A_{i-1} A_i \ldots A_j A_{j+1}) = S(A_{i-1}A_{j+1}).
 \end{equation}
 \end{proof}
 
 It follows directly that the entropy of the intermediate ancillae $A_i, \ldots, A_j$ vanishes when given the joint state of the ancillae $A_{i-1}$ and $A_{j+1}$,
 \begin{equation}
 \begin{split} 
     S(A_i \ldots A_j | A_{i-1} A_{j+1}) & = S(A_{i-1} A_i \ldots A_j A_{j+1}) \\
        & \,~~~ - S(A_{i-1}A_{j+1}) \\
        & = 0.
 \end{split}
 \end{equation}
 Evidently, if the state of the {\em boundary} of the chain is known, then the intermediate ancillae exist in a pure quantum superposition. The joint state of $A_i, \ldots, A_j$ is fully determined (a zero-entropy state), given the joint state of $A_{i-1}$ that measured $Q$ in the past, together with $A_{j+1}$ that measured $Q$ in the future. Thus, for measurements on unprepared quantum systems, the entropy of an arbitrarily-long ancilla chain is found only in its boundary. The entropy Venn diagram for the boundary and the bulk of the measurement chain is shown in Fig.~\ref{fig:chain-arbitrary}. 
 
 That the entropy of a chain of measurements is determined entirely by the entropy of the chain's boundary may seem remarkable, but is reminiscent of the holographic principle~\cite{tHooft1993,Susskind1995,SusskindWitten1998}. Indeed, it is conceivable that an extension of the one-dimensional quantum chains we discussed here to tensor networks~\cite{EvenblyVidal2011} could make this correspondence more precise~\cite{Swingle2012}. We contrast this result with the previous Theorem~\ref{theorem-cq-prepared} for measurements on prepared quantum systems, where the entropy resided only at the end of the chain since the preparation was already known. 
 \begin{figure} % Fig. 8
     \centering
     \includegraphics[width=\linewidth]{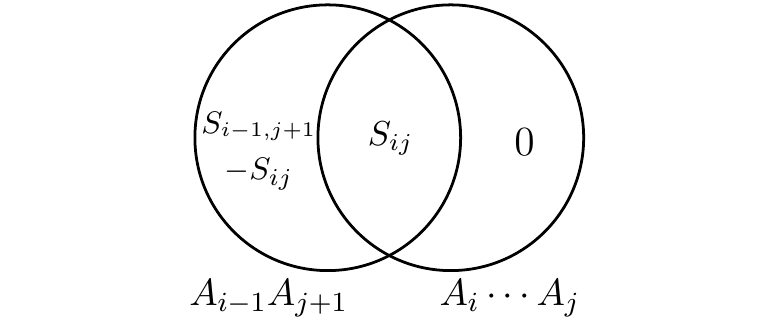} %venn_AiAj-nonmarkov-margin.pdf
     \caption{Entropy Venn diagram for an unamplified measurement sequence with ancillae $A_{i-1}, A_i , \ldots, A_j, A_{j+1}$. According to~\eqref{eqn:cq-arbitrary}, the entropy of all intermediate ancillae, $A_i , \ldots, A_j$, vanishes when given $A_{i-1}$ and $A_{j+1}$, since the entropy resides at the boundary of the chain. In this figure, we use the notation $S_{k\ell} = S(A_k A_{\ell})$ for the pairwise entropy of any two ancillae $A_k$ and $A_{\ell}$.}
     \label{fig:chain-arbitrary}
 \end{figure}

 \section{Markovian quantum measurements}
 \label{sec:markov}
 The non-Markovian measurements we have been discussing up to this point are potentially fragile: while the pointers can consist of many subsystems (even a macroscopic number), the entanglement they potentially display with other quantum systems will be lost even if only a single qudit escapes our control (and therefore,  mathematically speaking, must be traced over). In this section we discuss a second step within von Neumann's second stage of quantum measurement, where we observe the fragile quantum ancilla using a secondary observer. While this quantum ``observer of the observer" also potentially consists of many different subsystems, it is robust in the sense that tracing over any of the degrees of freedom making up the pointer variable does not affect the relative state of the pointer and the quantum system or other devices.

 \subsection{Amplifying quantum measurements}
 \label{sec:amplify}
 To amplify a measurement, we observe the first quantum observer (denoted by $A_1$) by measuring $A_1$ in the same basis with a detector $D_1$. This additional interaction with the first ancilla in~\eqref{QA_stageI} leads to the tripartite entangled state 
 \begin{equation} \label{QAD}
    |QA_1D_1\ra =  \mathbbm{1}_Q \otimes U_{A_1D_1} \, |QA_1\ra \, |0\ra = \!\sum_{x_1} \alpha^{(1)}_{x_1} \, |\wt{x}_1 x_1 x_1\ra.
 \end{equation}
 Tracing over the quantum system, we find that detector $D_1$ is perfectly correlated with the quantum ancilla $A_1$ according to the density matrix 
 \begin{equation} \label{AD}
    \rho(A_1D_1) = \sum_{x_1} q^{(1)}_{x_1} \, |x_1x_1\ra\la x_1x_1|,
 \end{equation}
 where $q^{(1)}_{x_1} = |\alpha^{(1)}_{x_1}|^2$. That is, they consistently reflect the same measurement outcomes. Together, $A_1$ and $D_1$ are still entangled with the quantum system. In Fig.~\ref{fig:amplify-venn} we show the entropy Venn diagrams for the entangled state~\eqref{QAD} and the correlated state~\eqref{AD}. Since the underlying state~\eqref{QAD} is pure, the ternary mutual entropy vanishes, $S(Q:A_1:D_1) = 0$. In other words, the correlations that are created between the devices (the $S(A_1:D_1)$ dits of information that are gained in the measurement) are not shared with the quantum system.
 \begin{figure} % Fig. 9
    \centering
    \includegraphics[width=\linewidth]{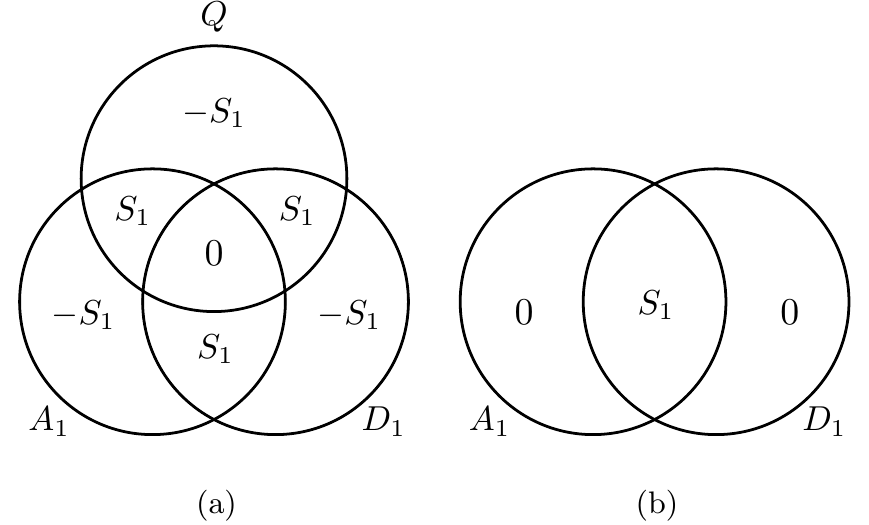} %venn_QAD-margin.pdf
    \caption{Entropy Venn diagram for (a) the tripartite entangled state~\eqref{QAD}. (b) Tracing over the quantum system, $A_1$ and $D_1$ are perfectly correlated as in Eq.~\eqref{AD}. The $S(A_1:D_1) = S_1$ bits of information gained in the measurement are not shared with the quantum system since the mutual ternary entropy vanishes, $S(Q:A_1:D_1) = 0$. The quantity $S_1 = H[ q^{(1)}] $ is the marginal entropy of each of the three subsystems, $Q$, $A_1$ and $D_1$.}
    \label{fig:amplify-venn}
 \end{figure}

 The macroscopic device $D_1$ is composed of many qudits $D^{(1)}_1, \ldots, D^{(n)}_1$ that all measure the quantum ancilla $A_1$ according to the sequence of entangling operations $U_{\!\!A_1 D_1^{(n)}} \ldots U_{\!\!A_1 D_1^{(1)}}$ (see Fig.~\ref{fig:amplify}). That is, Eq.~\eqref{QAD} can be expanded to
 \begin{equation} 
    |Q A_1 D_1\ra = \sum_{x_1} \, \alpha^{(1)}_{x_1} \, |\wt{x}_1 \ra  \,|x_1\ra \, |x_1\ra_{D^{(1)}_1} \ldots |x_1\ra_{D^{(n)}_1} \,.
 \end{equation} 
 The measurement outcome is read out from the state of the joint system
 \begin{equation} 
    \rho(D_1^{(1)} \!\ldots D_1^{(n)}) = \sum_{x_1} q^{(1)}_{x_1} \, |x_1 \ldots x_1 \ra\la x_1 \ldots x_1|,
 \end{equation}
 where it is clear that the device $D_1$ is self-consistent and all of its components reflect the same measurement outcome. This state is robust in the sense that it is not necessary to ``keep track'' of all qudits in the detector $D_1$ to observe correlations. Thus, tracing over any of the states in the expression above returns an equivalently self-consistent state.
 
 \begin{figure} % Fig. 10
    \centering
    \includegraphics[width=\linewidth]{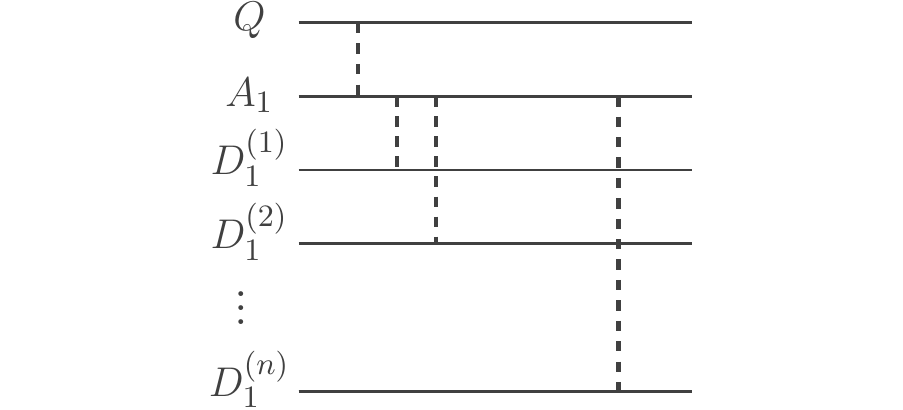} %measurement-2-margin.pdf
    \caption{Observing the quantum observer $A_1$ with a detector $D_1$. Dashed lines indicate the entanglement created by the measurement between the ancilla $A_1$ and each of the $n$ qudits $D_1^{(1)},\ldots, D_1^{(n)}$ that comprise the detector $D_1$.}
    \label{fig:amplify}
 \end{figure}
 
 In the following two sections, we amplify a chain of consecutive measurements of a prepared and an unprepared quantum system. Unlike our previous results for unamplified measurements, we will find that the joint state of detectors is now always classical (diagonal in the ancilla product basis), leading to entropy distributions that are significantly different from those of the unamplified ancillae.

 \subsection{Amplifying consecutive measurements of a prepared quantum state}
 \label{sec:markov-prepared}
 We begin by first considering the amplification of consecutive measurements of a prepared quantum state. Introducing a second pair of devices $A_2$ and $D_2$, Eq.~\eqref{QAD} evolves to
 \begin{equation}\label{QADAD}
    |QA_1D_1A_2D_2\ra = \sum_{x_1 x_2} \alpha^{(1)}_{x_1} \, U^{(2)}_{x_1x_2} \, |\wt{x}_2 \, x_1x_1 \, x_2x_2\ra.
 \end{equation} 
 Again, we find detector $D_2$ to be perfectly correlated with the quantum ancilla $A_2$. The joint state of the detectors $D_1$ and $D_2$ is the classical density matrix
 \begin{equation}\label{D1D2} 
    \rho(D_1D_2) = \sum_{x_1x_2} |\alpha^{(1)}_{x_1}|^2 \, |U^{(2)}_{x_1x_2}|^2 \, |x_1x_2\ra\la x_1x_2|.
 \end{equation} 
 This state is diagonal in the ancilla product basis, unlike the state~\eqref{AB-prepared} before amplification. Thus, the effect of amplifying the ancillae is a removal of all off-diagonal elements in the joint density matrices. 
 
 From~\eqref{D1D2}, we see that for repeated measurements in the same basis ($U^{(2)}_{x_1x_2} = \delta_{x_1x_2}$) the results are fully correlated: when $D_2$ measures in the same basis as $D_1$, the joint density matrix~\eqref{D1D2} reduces to $\rho(D_1D_2) = \sum_{x_1} |\alpha^{(1)}_{x_1}|^2 \, |x_1x_1\ra\la x_1x_1|$ so that the entropy of $D_2$ given $D_1$ vanishes, $S(D_2|D_1) = S(D_1D_2) - S(D_1) = 0$. The conditional probability to record the outcome $x_2$, given that the first measurement yielded $x_1$, is simply $p(x_2|x_1) = \delta_{x_1x_2}$. In other words, both devices agree on the outcome, as expected. It appears as if the quantum system had indeed collapsed into an eigenstate of the first device $D_1$ since the second device $D_2$ correctly confirms the measurement outcome. This result is consistent with the Copenhagen view of the quantum state during the measurement sequence as $|Q\ra \rightarrow |\wt{x}_1\ra \rightarrow |\wt{x}_1\ra$. However, we see that no collapse assumption is needed for a consistent description of the measurement process, and in fact, all amplitudes of the quantum system are preserved. That is, \eqref{QADAD} continues to evolve as a pure state. 
 
 In addition, the probability distribution for the second measurement with the pair $A_2D_2$ is consistent with a collapse postulate as it is given by the incoherent sum $q^{(2)}_{x_2} \!=\! \sum_{x_1} \! |\alpha^{(1)}_{x_1}|^2 \, |U^{(2)}_{x_1x_2}|^2$, instead of the coherent expression $\sum_{x_1} \! |\alpha^{(1)}_{x_1}\, U^{(2)}_{x_1x_2}|^2$, which is the result if the first measurement with $A_1D_1$ had never occurred.

 \subsection{Amplifying consecutive measurements of an unprepared quantum state}
 \label{sec:markov-unprepared}
 In this section, we study consecutive measurements of an unprepared quantum state, which will yield an entropy Venn diagram for the detectors that differs significantly from Fig.~\ref{fig:venn-ABC-nonmarkov} for the quantum ancillae. To begin, we follow the procedure introduced in Sec.~\ref{sec:amplify}, and amplify the state~\eqref{qabc} of three consecutive measurements of an unprepared quantum state. 
 
 First, we show that amplifying the qubits on the boundary of the chain of measurements does not affect the coherence of the joint state~\eqref{rhoABC}. Introducing macroscopic devices $D_1$ and $D_3$ that amplify the quantum ancillae $A_1$ and $A_3$, respectively, we find that the state~\eqref{qabc} evolves to
 \begin{widetext}
 \begin{equation}\label{QRADADAD}
    |QRA_1D_1A_2A_3D_3\ra \,=\,\frac1{\sqrt d}  \sum_{x_1x_2x_3}  \, U^{(2)}_{x_1x_2} \,\, U^{(3)}_{x_2x_3}  \,\, |\wt{x}_3 \, x_1 x_1 x_1 \, x_2 \, x_3 x_3\ra. 
 \end{equation}
 \end{widetext}
 As before, each pair of systems $A_i D_i$ are perfectly correlated and reflect the same outcome from their measurement of $Q$. Tracing over the density matrix formed from this wavefunction, we find that the new state of $A_1A_2A_3$ is unchanged from Eq.~\eqref{rhoABC}.

 In contrast, amplifying the {\em intermediate} ancilla destroys all of the coherence in the original state~\eqref{rhoABC}. That is, measuring $A_2$ with a detector $D_2$ leads to a fully incoherent density matrix for $A_1A_2A_3$ that is now equivalent to the joint state of detectors
 \begin{equation}\label{rhoA'B'C'}
    \rho(D_1D_2D_3) = \frac1d  \sum_{\substack{x_1x_2x_3}} \! |U^{(2)}_{x_1x_2}|^2  |U^{(3)}_{x_2x_3}|^2  |x_1x_2x_3\ra\la x_1x_2x_3|.
 \end{equation}

 We can contrast this state to the result we obtained for unamplified measurements in Eq.~\eqref{rhoABC} using entropy Venn diagrams. Compare the diagram in Fig.~\ref{fig:venn-ABC-markov} for the state $\rho(D_1D_2D_3)$ [Eq.~\eqref{rhoA'B'C'}] to the diagram in Fig.~\ref{fig:venn-ABC-nonmarkov} for the unamplified state $\rho(A_1A_2A_3)$ [Eq.~\eqref{rhoABC}]. Clearly, amplification of just the intermediate ancilla $A_2$ (or, equivalently, all three quantum ancillae) has destroyed the coherence of the original state $\rho(A_1A_2A_3)$, which was encoded in the $A_2$ subsystem. Note that pairwise entropies are the same for both amplified and unamplified measurements of unprepared quantum systems, e.g., $S(A_iA_j) = S(D_iD_j)$. We proved previously in Theorem~\ref{theorem-cq-unprepared} of Sec.~\ref{sec:coherence} that pairwise density matrices~\eqref{eqn:cq-arbitrary-pairwise} are always diagonal, so that amplifying those ancillae does not affect their joint density matrix.
 \begin{figure} % Fig. 11
    \centering
    \includegraphics[width=\linewidth]{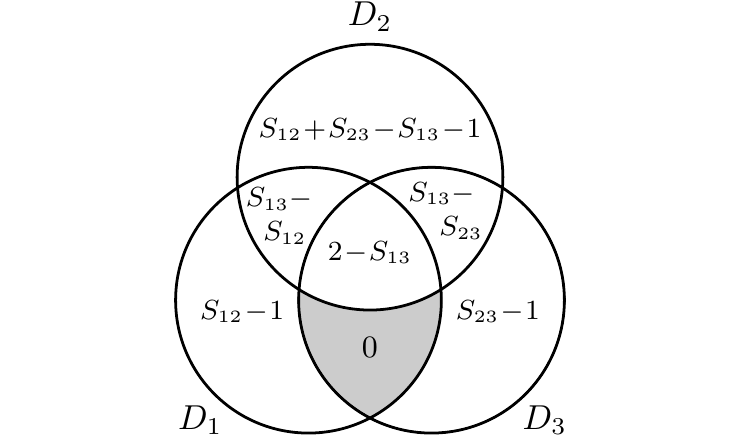} %venn_ABC-markov-margin.pdf
    \caption{Entropy Venn diagram for the sequence of detectors $D_1$, $D_2$, and $D_3$ that observe (amplify) the quantum ancillae $A_1$, $A_2$, and $A_3$, according to~\eqref{rhoA'B'C'}. Only amplification of the intermediate ancilla $A_2$ is sufficient to destroy the coherence in the original state~\eqref{rhoABC}. Note that all the pairwise entropies are unchanged by amplification, $S(A_iA_j) = S(D_iD_j) = S_{ij}$.}
  \label{fig:venn-ABC-markov}
 \end{figure}

 We apply these results to the case of qubit measurements ($d=2$), which are implemented with the rotation matrix in Eq.~\eqref{rotation}. For three consecutive measurements with $\theta_2 = \theta_3 = \pi/4$, the joint density matrix of all three detectors, which we show for comparison to the unamplified state~\eqref{three}, is diagonal:
  \begin{equation}\label{three-markov}
    \rho(D_1D_2D_3) = \frac{1}{8} 
    \begin{pmatrix}
    ~\mathbbm{1} & 0           & 0      	& 0 ~  \\
    ~0		 & \mathbbm{1} & 0           	& 0 ~  \\ 
    ~0           & 0           & \mathbbm{1}  	& 0 ~  \\
    ~0           & 0           & 0    		& \mathbbm{1}~ 
    \end{pmatrix}.
 \end{equation}
 As with the unamplified state~\eqref{three}, the pairwise entropies for the detectors are also 2 bits. However, the tripartite entropy has increased to $S(D_1D_2D_3) = 3$ bits from the 2 bits we found for $S(A_1A_2A_3)$. Compare the resulting entropy Venn diagram in Fig.~\ref{fig:venn-abc-qubit-markov} to the diagram in Fig.~\ref{fig:venn-abc-qubit-nonmarkov} obtained for unamplified qubit measurements.  
 
 The difference between the unamplified density matrix $\rho(A_1A_2A_3)$ in Eq.~\eqref{rhoABC} and the amplified state $\rho(D_1D_2D_3)$ in Eq.~\eqref{rhoA'B'C'} can be ascertained by revealing the off-diagonal terms via quantum state tomography (see, e.g.,~\cite{Whiteetal1999}), by measuring just a single moment~\cite{Tanaka2014} of the density matrix, such as ${\rm Tr} [\rho(A_1A_2A_3)^2]$, or else by direct measurement of the wavefunction~\cite{Lundeenetal2011}.

 \begin{figure} % Fig. 12
    \centering
    \includegraphics[width=\linewidth]{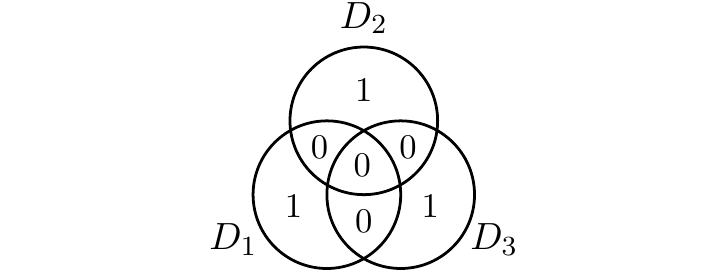} %venn_ABC-qubit-markov-margin.pdf
    \caption{Entropy Venn diagram for amplification with detectors $D_1$, $D_2$, and $D_3$ in~\eqref{three-markov} of three qubit ancillae that measured an unprepared quantum system. Ancilla $A_2$ measured $Q$ at $\theta_2=\pi/4$ relative to the basis of $A_1$, while $A_3$ measured $Q$ at $\theta_3=\pi/4$ relative to $A_2$. In this case, all three detectors are uncorrelated. Amplification of $A_2$ with $D_2$ alone is sufficient to destroy the coherence in~\eqref{three}.}
  \label{fig:venn-abc-qubit-markov}
 \end{figure}
 
 The results in the two preceding sections are compatible with the usual formalism for orthogonal measurements~\cite{Peres1995,Holevo2011}, where the conditional probability $p(x_2|x_1)$ to observe outcome $x_2$, given that the previous measurement yielded outcome $x_1$, is given by
 \begin{equation}\label{cond}
    p(x_2|x_1) = |U^{(2)}_{x_1x_2}|^2\;. 
 \end{equation}
 Indeed, our findings thus far are fully consistent with a picture in which a measurement collapses the quantum state (or alternatively, where a measurement recalibrates an observer's ``catalogue of expectations"~\cite{Schroedinger1935,Englert2013,Fuchsetal2014}). 
 
 To see this, we write the joint density matrix $\rho(D_1D_2)$, found by tracing~\eqref{rhoA'B'C'} over $D_3$, in the collapse picture. For a detector $D_1$ that records outcome $x_1$ with probability $1/d$ and a detector $D_2$ that measures the same quantum state (at an angle determined by the rotation matrix $U^{(2)}$), the resulting density matrix is
 \begin{equation}
    \wt{\rho}\,(D_1D_2)=\frac1d\sum_{x_1} |x_1\ra \la x_1| \otimes \rho^{x_1}_{D_2} \, ,
 \end{equation}
 where the state $\rho^{x_1}_{D_2}$ of $D_2$ is defined using the projection operators $P_{x_1} \!=\! |x_1\ra\la x_1|$ on the state of $D_1$,
 \begin{equation} 
    \rho^{x_1}_{D_2}\!=\! \frac{{\rm Tr}_{D_1}\!\Big[P_{x_1}\, \rho(D_1D_2) P_{x_1}^\dagger\Big]}{ {\rm Tr}_{D_1D_2}\!\left[\!P_{x_1} \, \rho(D_1D_2)P_{x_1}^\dagger\!\right]} =\sum_{x_2} |U^{(2)}_{x_1x_2}|^2 \, |x_2\ra\la x_2|.
 \end{equation}
 In other words, the state $\rho(D_1D_2)$ that was obtained in a unitary formalism is equivalent to the collapse version $\wt{\rho}\,(D_1D_2)$. However, despite these consistencies with the collapse picture, we emphasize that the actual measurements induce no irreversible collapse and that all amplitudes in the underlying pure-state wavefunction~\eqref{QRADADAD} are preserved and evolve unitarily throughout the measurement process.

 \subsection{Quantum Markov chains}
 \label{sec:markov-chain-proof}
 One of the key differences between the entropy Venn diagrams in Figs.~\ref{fig:venn-ABC-nonmarkov} and~\ref{fig:venn-ABC-markov} is the vanishing conditional mutual entropy~\cite{CerfAdami1998} for amplified measurements, $S(D_1:D_3|D_2) = 0$. Before amplification, the equivalent quantity for the quantum ancillae is in general non-zero, $S(A_1:A_3|A_2) \ge 0$. Evidently, the intermediate measurement with $D_2$ has, from the perspective of $D_2$ (meaning, given the state of $D_2$) {\em erased all correlations} between the first detector $D_1$ and the last detector $D_3$ in the measurement sequence. The vanishing of the conditional mutual entropy is precisely the condition that is fulfilled by quantum Markov chains as we will outline below.
 
 Using the results for unprepared quantum states (this holds equally for prepared quantum states), we demonstrate that the chain of detectors, $D_1, D_2, D_3$, which consecutively measured a quantum system is Markovian, as defined in~\cite{HaydenMarkov2004} (see also~\cite{Datta2015} and references therein). We prove later in this section in Theorem~\ref{theorem-markov} that this result can be extended to any number of consecutive measurements, not just three. To show that $S(D_1\!:\!D_3|D_2)$ is indeed zero, we compute the joint entropy $S(D_1D_2D_3)$ of all three detectors. From Eq.~\eqref{rhoA'B'C'}, we find
 \begin{equation}
 \begin{split}\label{entropy_A'B'C'}
    S(D_1D_2D_3) = 1 & -  \frac{1}{d} \sum_{x_1x_2} |U^{(2)}_{x_1x_2}|^2  \log_d |U^{(2)}_{x_1x_2}|^2 \\
                   & - \frac{1}{d}  \sum_{x_2x_3}  |U^{(3)}_{x_2x_3}|^2  \log_d |U^{(3)}_{x_2x_3}|^2 ,
 \end{split} 
 \end{equation}
 or, $S(D_1D_2D_3) = S(D_1) + S(D_2|D_1) + S(D_3|D_2)$. However, using the chain rule for entropies~\cite{CerfAdami1998}, the tripartite entropy can also be written generally in the form $S(D_1D_2D_3) = S(D_1) + S(D_2|D_1) + S(D_3|D_2D_1)$. From these two expression, we see immediately that
 \begin{equation}
        S(D_3|D_2D_1) = S(D_3|D_2).
 \end{equation}
 Thus, the entropy of the detector $D_3$ is not reduced by conditioning on more than the state of the previous detector $D_2$. This is the Markov property for entropies~\cite{HaydenMarkov2004,Datta2015}. 
 
 The Markov property further implies that detectors $D_1$ and $D_3$ are independent from the perspective of $D_2$, since the conditional mutual entropy~\cite{CerfAdami1998} vanishes (see the grey region in Fig.~\ref{fig:venn-ABC-markov}),
 \begin{equation}
    S(D_1\!:\!D_3|D_2) =  S(D_3|D_2) - S(D_3|D_2D_1) = 0 \,.
 \end{equation} 
 This result is consistent with the notion that the measurement with $D_2$ collapsed the state of the wavefunction, erasing any (conditional) information that detector $D_3$ could have had about the prior measurement with $D_1$. The conditional mutual entropy does not vanish for unamplified measurements, $S(A_1:A_3|A_2) \ge 0$, reflecting the fundamentally non-Markovian nature of the chain of quantum ancillae. In other words, as long as the measurement chain remains unamplified (for example, the $A_2$ subsystem in~\eqref{rhoABC}), the intermediate measurement does not erase the correlations between $A_1$ and $A_3$ (compare the gray region in Fig.~\ref{fig:venn-ABC-markov} to the same region in Fig.~\ref{fig:venn-ABC-nonmarkov}). 
 
 We now provide a formal proof of the statement that the chain of detectors that amplified the quantum ancillae is equivalent to a quantum Markov chain.

 %%%%%%%%%%%%%%%%%%%%%
 % Theorem  markov   %
 %%%%%%%%%%%%%%%%%%%%%
 \begin{thm}\label{theorem-markov}
 A set of consecutive quantum measurements is non-Markovian until it is amplified. Specifically, the sequence of devices $D_i, \ldots ,D_j$, with $i < j$, that measure (amplify) the quantum ancillae $A_i,\ldots, A_j$ (which themselves measured a quantum system $Q$) forms a quantum Markov chain:
  \begin{equation}
      S(D_j|D_{j-1} \ldots D_i) = S(D_j|D_{j-1}).
  \end{equation} 
 \end{thm}

 \begin{proof}
 We first show that the Markov property of probabilities implies the Markov property for entropies (see, e.g., Refs.~\cite{HaydenMarkov2004,Datta2015}). If consecutive measurements on a quantum system can be modeled as a Markov process, the probability to observe outcome $x_j$ in the $j$th detector, conditional on previous measurement outcomes, depends only on the last outcome $x_{j-1}$,
 \begin{equation}\label{eqn:markov}
    p(x_j|x_{j-1} \ldots x_i) = p(x_j|x_{j-1}) . 
 \end{equation}
 Inserting Eq.~(\ref{eqn:markov}) into the expression for the conditional entropy~\cite{CerfAdami_PRL1997} gives 
 \begin{equation}\begin{split}
    S(D_j|D_{j-1} \!\ldots D_i)\! & =\! - \!\!\!\! \sum_{\substack{ x_i \ldots x_j }} \!\!\! p(x_i \ldots x_j) \log_d  p(x_j|x_{j-1} \!\ldots x_i)  \\
                          & =\! - \!\!\!\! \sum_{\substack{ x_i \ldots x_j }} \!\!\! p(x_i \ldots x_j)  \log_d  p(x_j|x_{j-1}) \,.
 \end{split}\end{equation}
 A partial summation over the joint probability distribution gives 
 \begin{equation}
    p(x_{j-1}x_j) \,= \!\!\!\!\! \sum_{x_i\cdots x_{j-2}} \!\!\! p(x_i \ldots x_j),
 \end{equation} 
 so that the entropic condition satisfied by a quantum Markov chain is
 \begin{equation}\begin{split}\label{eqn:MarkovEntropy}
    \!S(D_j|D_{j-1} \ldots D_i) & =  - \!\!\! \sum_{x_{j-1}x_j} \!\! p(x_{j-1}x_j )  \log_d p(x_j|x_{j-1}) \\
                          & = S(D_j|D_{j-1}) , .
 \end{split}\end{equation}
 
 We now show that the chain of amplified measurements satisfies the entropic Markov property~\eqref{eqn:MarkovEntropy}. For $n$ consecutive measurements, the state $|\Psi\ra = |Q A_1 \ldots A_n\ra$ of $Q$ and all ancillae is given by
 \begin{equation}
        |\Psi\ra = \!\!\!\sum_{x_1 \cdots x_n} \!\!\alpha^{(i)}_{x_i} \, U^{(2)}_{x_1x_2} \ldots U^{(n)}_{x_{n-1}x_n} \, |\wt{x}_n \, x_1 \ldots x_n\ra.
 \end{equation}
 After amplifying this state, we find that the density matrix for the joint set of sequential detectors, $D_i, \ldots, D_j$, with $i < j$, is diagonal, as expected,
 \begin{equation}\label{eqn:jointcollapse}\begin{split}
    \rho(D_i \ldots D_j)& =  \sum_{\substack{ x_i }} q^{(i)}_{x_i} \, |x_i\ra\la x_i| \\
      & \, \otimes  \sum_{\substack{ x_{i+1} }} |U^{(i+1)}_{x_i x_{i+1}}|^2 \, |x_{i+1}\ra\la x_{i+1}|  \\
      \cdots & \, \otimes \sum_{\substack{ x_j }} |U^{(j)}_{x_{j-1} x_j}|^2 \, |x_j \ra \la x_j |.
 \end{split}\end{equation}
 The probability distribution $q^{(i)}_{x_i}$ of the $i$th device can be obtained from~\eqref{diagonalize-prepared-general-normalization}. The entropy of~\eqref{eqn:jointcollapse} is 
 \begin{equation}\label{eqn:collapse_entropy_n} 
 \begin{split} 
    \hspace{-0.3cm} S(D_i \ldots D_j) \!  = & - \!\!\!\!\!\! \sum_{\substack{ x_i \ldots x_{j-1} }} \!\!\!\!\! \Big(q^{(i)}_{x_i} \,|U^{(i+1)}_{x_i x_{i+1}}|^2  \ldots |U^{(j-1)}_{x_{j-2} x_{j-1}}|^2 \Big)\\
	& ~~\times \log_d \! \Big( q^{(i)}_{x_i}\, |U^{(i+1)}_{x_i x_{i+1}}|^2 \ldots |U^{(j-1)}_{x_{j-2} x_{j-1}}|^2\Big) \\
	& - \!\!\!\! \sum_{\substack{ x_{j-1} x_j }}  \!\!\! q^{(j-1)}_{x_{j-1}}\, |U^{(j)}_{x_{j-1} x_j}|^2 \log_d |U^{(j)}_{x_{j-1} x_j}|^2, \!
 \end{split} 
 \end{equation}
 where $q^{(j-1)}_{x_{j-1}}$ is the probability distribution of $D_{j-1}$. The first term in Eq.~\eqref{eqn:collapse_entropy_n} is just the joint entropy $S(D_i \ldots D_{j-1})$, so that the entropy of the $j$th detector, conditional on the previous detectors, is
 \begin{equation}\label{eqn:collapse_cond1}
 \begin{split}
    S( D_j | D_{j-1} \!\ldots D_i) & \!= S(D_i \ldots D_j) - S(D_i \ldots D_{j-1}) \\
	  & \!=\! - \!\!\!\! \sum_{\substack{ x_{j-1} x_j }} \!\!\! q^{(j-1)}_{x_{j-1}} |U^{(j)}_{x_{j-1} x_j}|^2 \log_d |U^{(j)}_{x_{j-1} x_j}|^2\!. 
 \end{split}
 \end{equation}

 All that remains is to show that~\eqref{eqn:collapse_cond1} is equal to $S(D_j|D_{j-1})$. A simple calculation using the density matrix for two amplified consecutive measurements with $D_{j-1}$ and $D_j$,
 \begin{equation} 
    \rho(D_{j-1}D_j) = \!\!\! \sum_{x_{j-1}x_j} \! q^{(j-1)}_{x_{j-1}} \,  |U^{(j)}_{x_{j-1}x_j}|^2 \, |x_{j-1}x_j\ra\la x_{j-1}x_j|,
 \end{equation} 
 yields the joint entropy,
 \begin{equation}
 \begin{split}
    S(D_{j-1} D_j) =  & - \! \sum_{\substack{ x_{j-1} }} q^{(j-1)}_{x_{j-1}} \log_d  q^{(j-1)}_{x_{j-1}}  \\
                   & - \!\!\! \sum_{\substack{ x_{j-1} x_j }} \!\!  q^{(j-1)}_{x_{j-1}} |U^{(j)}_{x_{j-1} x_j}|^2 \log_d |U^{(j)}_{x_{j-1} x_j}|^2 .
 \end{split}
 \end{equation} 
 The first term in this expression is the entropy of $D_{j-1}$ (all marginal density matrices and entropies are the same for amplified and unamplified ancillae; this is proved formally later in Lemma~\ref{lemma-AiDi} of Sec.~\ref{sec:info-about-quantum-system}),
 \begin{equation}
 S(D_{j-1}) = H[q^{(j-1)}] = - \sum_{\substack{ x_{j-1}}}  q^{(j-1)}_{x_{j-1}}  \log_d q^{(j-1)}_{x_{j-1}}.
 \end{equation} 
 The conditional entropy $S(D_j|D_{j-1})$ is thus
 \begin{equation}
 \begin{split}
    S(D_j|D_{j-1}) & = S(D_{j-1}D_j) - S(D_{j-1}) \\
              & = -  \!\!\!\! \sum_{\substack{ x_{j-1} x_j }}  \!\! q^{(j-1)}_{x_{j-1}} |U^{(j)}_{x_{j-1} x_j}|^2 \log_d |U^{(j)}_{x_{j-1} x_j}|^2, \!
 \end{split}
 \end{equation}
  which is the same as~\eqref{eqn:collapse_cond1}.
 \end{proof}
 We emphasize that the result that amplified measurements are Markovian holds for measurements of unprepared as well as prepared quantum states. 
  
 %%%%%%%%%%%%%%%%%%%%%%
 % Corollary - Markov %
 %%%%%%%%%%%%%%%%%%%%%%
 \begin{cor}\label{corollary-markov}
 The Markovian nature of amplified measurements implies that the detectors $D_{i}$ and $D_{j}$ share no entropy (are independent) from the perspective of the intermediate detectors, $D_{i+1}, \ldots, D_{j-1}$, since the conditional mutual entropy vanishes:
 \begin{equation}
      S(D_{i} : D_{j}|D_{i+1} \ldots D_{j-1})=0.
 \end{equation} 
 \end{cor}

 \begin{proof}
 The conditional mutual entropy is defined~\cite{CerfAdami1998} as a difference between two conditional entropies,
 \begin{equation}\begin{split} 
      S(D_{i} : D_{j}|D_{i+1} \ldots D_{j-1}) & = S(D_{j}|D_{j-1} \ldots D_{i+1}) \\
       & -S(D_{j}|D_{j-1} \ldots D_{i}). 
  \end{split}\end{equation}
 From Theorem~\ref{theorem-markov}, the two quantities on the right hand side of this expression are both equal to $S(D_j|D_{j-1})$. Therefore the conditional mutual entropy vanishes~\cite{HaydenMarkov2004}.
 \end{proof}
 
 For three detectors, the Markov property is
 \begin{equation} 
        S(D_{i-1} : D_{i+1}|D_i) = S(D_{i+1}|D_i) - S(D_{i+1}|D_iD_{i-1}) = 0. 
 \end{equation}        
 We see that, from the strong subadditivity (SSA) of quantum entropy~\cite{Lieb1973a,Lieb1973b}, 
 \begin{equation} 
    S(D_{i+1}|D_iD_{i-1}) \le S(D_{i+1}|D_i),
 \end{equation} 
 amplified measurements satisfy SSA with equality.

 The previous theorem established that the sequence of amplified measurements is a quantum Markov chain. Now, we will demonstrate that unamplified measurements are non-Markovian. In the following calculation, we use the state~\eqref{eqn:UnpreparedFullWavefunction} for measurements of unprepared quantum states for simplicity. We will find that the Markov property~\eqref{eqn:MarkovEntropy} is violated in this case, so that in general unamplified measurements are non-Markovian. 
 
 First, consider the joint density matrix for the sequence of quantum ancillae $A_i, \ldots, A_j$ (with $i < j$), similarly to~\eqref{eqn:UnpreparedFullWavefunction}. As in Eq.~\eqref{eqn:cq-arbitrary}, we find
 \begin{equation} \label{eqn:rho_i_to_j}
        \rho(A_i \!\ldots A_j) \!=\! \frac{1}{d} \! \sum_{x_i x_j}  \!p^{(ij)}_{x_ix_j}   |x_i\ra\la x_i| \otimes |\phi_{x_ix_j}\ra\la \phi_{x_ix_j}| \otimes |x_j\ra\la x_j| ,
 \end{equation}
 where the coefficients $p^{(ij)}_{x_ix_j} = |\beta^{(ij)}_{x_ix_j}|^2$ and the normalized, but non-orthogonal states $|\phi_{x_i x_j}\ra$ were defined in Eq.~\eqref{diagonalize-unprepared-general}. The joint states $|x_i \, \phi_{x_ix_j} \, x_j \ra$ are orthonormal, so the entropy of Eq.~\eqref{eqn:rho_i_to_j} is simply 
  \begin{equation}\label{eqn:entropy1}
        S(A_i \ldots A_j) = 1 - \frac{1}{d} \sum_{x_i x_j} p^{(ij)}_{x_ix_j} \log_d p^{(ij)}_{x_ix_j}.
 \end{equation} 
 The coefficients $p^{(ij)}_{x_ix_j}$ can be equivalently expressed in terms of $U^{(j)}$ as
 \begin{equation}
    p^{(ij)}_{x_ix_j} = |\beta^{(ij)}_{x_ix_j}|^2 = \sum_{x_{j-1}} p^{(i,j-1)}_{x_ix_{j-1}} ~  |U^{(j)}_{x_{j-1}x_j}|^2.
 \end{equation} 
 Inserting this into~\eqref{eqn:entropy1} and using the log-sum inequality~\footnote{The log-sum inequality~\protect{\cite{Cover2012}} states that for non-negative numbers $a_1,\ldots,a_d$ and $b_1,\ldots,b_d$,
 {\protect \[
   \sum_{x_i=1}^d a_{x_i} \log \frac{a_{x_i}}{b_{x_i}} \ge \left( \sum_{{x_i}=1}^d a_{x_i} \right) \log \frac{\sum_{{x_i}=1}^d a_{x_i}}{\sum_{{x_i}=1}^d b_{x_i}} ~,\]
 }
 with equality if and only if $a_{x_i}/b_{x_i} = $ const.} with $b_{x_{j-1}}=1$ and $a_{x_{j-1}} = p^{(i,j-1)}_{x_i x_{j-1}} \, |U^{(j)}_{x_{j-1}x_j}|^2$, we find that the joint entropy is bounded from below by
 \begin{equation}
 \begin{split}\label{eqn:entropy2}
    S(A_i \ldots A_j) \ge & - \frac{1}{d} \sum_{x_i x_{j-1}} p^{(i,j-1)}_{x_ix_{j-1}} \log_d p^{(i,j-1)}_{x_ix_{j-1}} \\
        & - \frac{1}{d} \sum_{x_{j-1} x_j} |U^{(j)}_{x_{j-1}x_j}|^2 \log_d |U^{(j)}_{x_{j-1}x_j}|^2.
 \end{split}
 \end{equation}
 The first term on the right hand side of Eq.~\eqref{eqn:entropy2} is simply $S(A_i \ldots A_{j-1}) - 1$, while the second term is $S(A_{j-1}A_j) - 1$. Given that $S(A_{j-1}) = 1$, it is straightforward to show that Eq.~\eqref{eqn:entropy2} can be rewritten as a difference between two conditional entropies,
 \begin{equation}
 \label{eqn:entropy3}
   S(A_j|A_{j-1}) - S(A_j|A_{j-1} \ldots A_i ) \le 1,
 \end{equation}
 with equality only when $p^{(i,j-1)}_{x_i x_{j-1}} \, |U^{(j)}_{x_{j-1}x_j}|^2$ is a constant. This occurs when $|U^{(j)}_{x_{j-1}x_j}|^2 = 1/d$ and $|U^{(\ell)}_{x_{\ell-1}x_\ell}|^2 = 1/d$ for one or more of the $\ell = i+1, \ldots, j-1$ matrices. This shows that conditioning on more than just the state of the last ancilla $A_{j-1}$ will reduce the conditional entropy of ancilla $A_j$ (by at most 1). Since Eq.~\eqref{eqn:entropy3} is not equal to zero in general, we conclude that the sequence of unamplified measurements is non-Markovian.
 
\section{Effects of amplifying quantum measurements}
\label{sec:five}
 In the previous sections~\ref{sec:three} and~\ref{sec:markov}, we focused on consecutive measurements of a quantum system and discussed the concepts of non-Markovian (unamplified) and Markovian (amplifiable) sequences, respectively. It is reasonable to ask whether there are entropic relationships between those two kinds of measurements. Introducing a second step to von Neumann's second stage serves precisely to establish such relationships. In this section, we establish the following three properties: Markovian detectors carry less information about the quantum system than non-Markovian devices; the shared entropy between consecutive non-Markovian devices is larger than the respective quantity for amplified measurements; the last Markovian detector in a quantum chain is inherently more random than its non-Markovian counterpart, given the combined results of all previous measurements.

 \subsection{Information about the quantum system}
 \label{sec:info-about-quantum-system}
 We first calculate how much information about the quantum system, $Q$, is encoded in the last device in a chain of consecutive measurements of $Q$. To do this, we prove two Lemmas that state that the marginal entropy of the quantum system is always equal to the entropy of the last ancilla in the chain of measurements, and that the marginal entropy of a quantum ancilla is unaffected by amplification.

 %%%%%%%%%%%%%%%
 %  Lemma - Q  %
 %%%%%%%%%%%%%%% 
 \begin{lem}\label{lemma-Q}
 The entropy of the quantum system, $Q$, is equal to the entropy of the last ancilla, $A_n$, in the chain of measurements:
 \begin{equation}
    S(Q) = S(A_n).
 \end{equation} 
 \end{lem}

 \begin{proof}
 Consider a series of consecutive measurements on a quantum system, $Q$, with $n$ ancillae. In general, following the measurements, the joint state of the quantum system and all ancillae $|\Psi\ra = |Q A_1 \ldots A_n \ra$ is given by the pure state [see also Eq.~\eqref{eqn:PreparedFullWavefunction}]
 \begin{equation}\label{eqn:FullWavefunction}
    |\Psi\ra = \!\!\!\sum_{ x_1 \ldots x_n } \!\! \alpha^{(1)}_{x_1} ~ U^{(2)}_{x_1 x_2}  \ldots  U^{(n)}_{x_{n-1} x_n} ~ |\wt{x}_n \,\,x_1 \ldots x_n \ra .
 \end{equation}
 The density matrix for the quantum system is found by tracing out all ancilla states from the full density matrix associated with~\eqref{eqn:FullWavefunction},
  \begin{equation}
  \label{rho-Q}
        \rho(Q) = \mathrm{Tr}_{A_1 \ldots A_n} (|\Psi\ra\la\Psi|) = \sum_{x_n} q^{(n)}_{x_n} ~ |\wt{x}_n\ra\la \wt{x}_n|,
  \end{equation} 
  where $q^{(n)}_{x_n}$ is the probability distribution for the last ancilla $A_n$ that can be obtained generally from Eq.~\eqref{diagonalize-prepared-general-normalization}. Clearly, \eqref{rho-Q} is equivalent to the density matrix for the last ancilla, and so the corresponding entropies are the same: $S(Q) = S(A_n) = H[q^{(n)}]$. An alternative proof is to note that a Schmidt decomposition of the pure state $|\Psi\ra\la\Psi|$ implies that $S(Q) = S(A_1 \ldots A_n)$. And, by Theorem~\ref{theorem-chain} (see Sec.~\ref{sec:info-past-measurements}), $S(A_n) = S(A_1 \ldots A_n)$, so that $S(Q) = S(A_n)$.
 \end{proof}

 %%%%%%%%%%%%%%%%%%%%%%%
 %  Lemma - Ai and Di  %
 %%%%%%%%%%%%%%%%%%%%%%% 
 \begin{lem}\label{lemma-AiDi}
 The entropy of a quantum ancilla, $A_i$, is unchanged if it is measured by an amplifying detector, $D_i$, so that for all $i$ in the chain of measurements:
 \begin{equation} 
    S(A_i) = S(D_i).
 \end{equation}
 \end{lem}

 \begin{proof}
 Amplifying the $i$th ancilla $A_i$ in~\eqref{eqn:FullWavefunction} with a detector $D_i$ yields the joint density matrix for $A_i$ and $D_i$,
 \begin{equation}
        \rho(A_iD_i) = \sum_{x_i} q^{(i)}_{x_i} ~ |x_i x_i \ra\la x_i x_i |,
 \end{equation}
 where $q^{(i)}_{x_i}$ is the probability distribution for $A_i$, as defined in~\eqref{diagonalize-prepared-general-normalization}. The two subsystems are perfectly correlated so that the density matrix and marginal entropy of $A_i$ is equivalent to $D_i$: $S(D_i) = S(A_i) = H[q^{(i)}]$.
 \end{proof}
 
 In the remaining sections, we will use the shortened notation $S(A_i) = S(D_i) = S_i$ for the marginal entropies. Using Lemmas~\ref{lemma-Q} and~\ref{lemma-AiDi}, we are now ready to prove the first theorem regarding information about the quantum system.

 %%%%%%%%%%%%%%%%%%%%%%%%%%%%%%%
 % Theorem - last measurement  %
 %%%%%%%%%%%%%%%%%%%%%%%%%%%%%%%
 \begin{thm}\label{theorem-last-info}
 The information that the last device in a series of measurements has about the quantum system is reduced when the measurements are amplified. That is,
 \begin{equation}\label{eqn:reduced-info-Q}
        S(Q:D_n) \le S(Q:A_n),
 \end{equation}
 for $n$ consecutive measurements of a prepared quantum state, $Q$.
 \end{thm}

 \begin{proof} 
 We start with the state~\eqref{eqn:FullWavefunction} for an unamplified chain of consecutive measurements of a prepared quantum state, $Q$, with $n$ ancillae. Tracing out all previous ancilla states from~\eqref{eqn:FullWavefunction}, the joint density matrix for the quantum system and the last ancilla is
 \begin{equation} \label{eqn:QAn} 
        \rho(QA_n) = \!\!\!\!\!\!\sum_{x_{n-1}x_nx'_n}\!\!\!\!\! q^{(n-1)}_{x_{n-1}} \,  U^{(n)}_{x_{n-1}x_n} \, U^{(n)\,*}_{x_{n-1}x'_n} \, |\wt{x}_n \, x_n \ra\la \wt{x}'_n \, x'_n |,
 \end{equation} 
 where $q^{(n-1)}_{x_{n-1}}$ is $A_{n-1}$'s probability distribution.
 
 If we amplify the measurement chain (or, equivalently, just the last measurement) the state~\eqref{eqn:QAn} becomes diagonal. That is,
 \begin{equation} \label{eqn:QDn}
        \rho(QD_n) = \!\!\!\!\sum_{x_{n-1}x_n}\!\! q^{(n-1)}_{x_{n-1}} \,\,  |U^{(n)}_{x_{n-1}x_n}|^2  \, |\wt{x}_n \, x_n \ra\la \wt{x}_n \, x_n |.
 \end{equation}
 Note that the amplification is equivalent to a completely dephasing channel~\cite{lloyd1997,bennett1997,horodecki2005} since we can write
 \begin{equation}
        \rho(QD_n) = \sum_{x_n} P_{x_n} \,\, \rho(QA_n) \, P_{x_n},
 \end{equation}
 where $P_{x_n} = |x_n\ra\la x_n|$ are projectors on the state of $A_n$. In other words, $\rho(QD_n)$ is formed from the diagonal elements of $\rho(QA_n)$.

 To show that the amplified mutual entropy is reduced as in Eq.~\eqref{eqn:reduced-info-Q}, it is sufficient to show that the joint entropy is increased. The mutual entropy for two subsystems is defined~\cite{CerfAdami1998} as $S(Q:A_n) = S(Q) + S(A_n) - S(QA_n)$ and similarly for $S(Q:D_n)$. Since, by Lemma~\ref{lemma-AiDi}, the marginal entropies are unchanged by the amplification, $S(A_n) = S(D_n)$, we have
 \begin{equation} 
        S(Q:D_n) = S(Q:A_n) + S(QA_n) - S(QD_n). 
 \end{equation}
 Therefore, we just need to show that $S(QD_n) \ge S(QA_n)$, which is easiest by considering the relative entropy of coherence~\cite{baumgratz2014,xi2015}. This quantity, $C_{\rm rel.ent.}(\rho) = S(\rho_{\rm diag})- S(\rho)$, is the difference between the entropies of a density matrix $\rho$ and a matrix $\rho_{\rm diag}$ that is formed from the diagonal elements of $\rho$. It is derived by minimizing the relative entropy $S(\rho \,\Vert\, \delta) = {\rm Tr}(\rho \log \rho - \rho \log \delta)$ over the set of incoherent matrices $\delta$. By Klein's inequality, the relative entropy is non-negative so that $S(\rho_{\rm diag}) \ge S(\rho)$, with equality if and only if $\rho$ is an incoherent matrix. In our case, $\rho$ and $\rho_{\rm diag}$ are given by $\rho(Q:A_n)$ and $\rho(Q:D_n)$, respectively. Therefore, it follows that $S(QD_n) \ge S(QA_n)$ and
 \begin{equation}
     S(Q:D_n) \le S(Q:A_n),
 \end{equation} 
 with equality if and only if $\rho(QA_n)$ is already diagonal in the ancilla product basis.
 \end{proof}
 
 To directly compute the mutual entropies in Theorem~\ref{theorem-last-info}, we first diagonalize the density matrix~\eqref{eqn:QAn} with the orthonormal states $|\Phi_{x_{n-1}}\ra = \sum_{x_n} U^{(n)}_{x_{n-1}x_n} \, |\wt{x}_n \, x_n \ra$, so that
 \begin{equation}
        \rho(QA_n) = \sum_{x_{n-1}} q^{(n-1)}_{x_{n-1}} \,\, |\Phi_{x_{n-1}}\ra\la\Phi_{x_{n-1}}|.
 \end{equation}
 The joint entropy of this state is simply the marginal entropy of $A_{n-1}$. That is, $S(QA_n) = S(A_{n-1}) = S_{n-1}$, which can also be derived using the Schmidt decomposition and the results of Theorem~\ref{theorem-chain} (see Sec.~\ref{sec:info-past-measurements}). Thus, using Lemma~\ref{lemma-Q}, the information that the last ancilla has about the quantum system is
 \begin{equation}\label{info-last-ancilla}
        S(Q:A_n) = 2 S_n - S_{n-1}.
 \end{equation}

 If we now amplify the measurement chain (or, equivalently, just the last measurement) the information that $D_n$ has about $Q$ will be reduced from~\eqref{info-last-ancilla}. From Eq.~\eqref{eqn:QDn}, the joint density matrix of $Q$ and $D_n$ can also be written as
 \begin{equation} \label{eqn:QDn-v2}
        \rho(QD_n) = \sum_{x_n} q^{(n)}_{x_n} \, |\wt{x}_n \, x_n \ra\la \wt{x}_n \, x_n |,
 \end{equation} 
 which leads to $S(QD_n) = S(D_n) = S_n$. Therefore, amplifying the measurement reduces the quantity~\eqref{info-last-ancilla} to
 \begin{equation}\label{info-last-detector}
    S(Q:D_n) = S_n,
 \end{equation} 
 where we used Lemmas~\ref{lemma-Q} and~\ref{lemma-AiDi} to write $S(Q) = S(A_n) = S(D_n) = S_n$. This quantity depends explicitly on only the last measurement, unlike~\eqref{info-last-ancilla}, which depends on the last two. The amount of information that the last device has about the quantum system before amplification, \eqref{info-last-ancilla}, and after, \eqref{info-last-detector}, is related by
 \begin{equation} 
    S(Q:D_n) = S(Q:A_n) + S_{n-1} - S_n\;.
 \end{equation} 
 Thus, the marginal entropies in a chain of consecutive measurements never decrease, $S_n \ge S_{n-1}$, since $S(Q:D_n) \le S(Q:A_n)$. The entropy Venn diagrams for the devices $A_n$ and $D_n$ and the quantum system are shown in Fig.~\ref{fig:info-last-measurement}. 
 \begin{figure} % Fig. 13
    \centering
    \includegraphics[width=\linewidth]{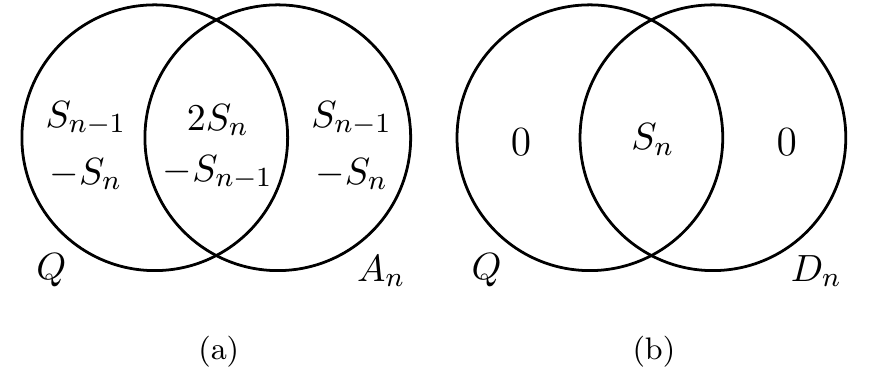} %venn_QAn_QDn-margin.pdf
    \caption{The entropy Venn diagrams for (a) the quantum system and the unamplified ancilla according to Eq.~\eqref{eqn:QAn}, and (b) the quantum system and the amplifying detector according to Eq.~\eqref{eqn:QDn-v2}. The information that the last device has about the quantum system is reduced when the measurement is amplified. That is, $S(Q:D_n) \le S(Q:A_n)$.}
    \label{fig:info-last-measurement}
 \end{figure}

 We can illustrate this loss of information about the quantum system by considering consecutive qubit measurements. Suppose that ancilla $A_{n-1}$ measures $Q$ at an angle $\theta_{n-1} = 0$ relative to $A_{n-2}$ and that $A_n$ measures $Q$ at an angle $\theta_n = \pi/4$ relative to $A_{n-1}$. In this case, the marginal entropies are $S_{n-1} = S_{n-2} = H[q^{(n-2)}]$ and $S_n = 1$ bit. The last detector, $D_n$, has one bit of information about quantum system, which is less than that of the unamplified ancilla: $S(Q:A_n) = 2 - H[q^{(n-2)}] \ge 1$. Interestingly, how much we know about the state of $Q$ prior to amplification is controlled by the entropy of an ancilla, $A_{n-2}$, located two steps down the measurement chain.

 \subsection{Information about past measurements}
 \label{sec:info-past-measurements}
 We now calculate how much information is encoded in a measurement device about the state of the measurement device that just preceded it in the quantum chain. In particular, we will show that the shared entropy $S(A_n:A_{n-1})$ between the last two devices in the measurement chain is reduced by the amplification process so that $S(D_n:D_{n-1}) \le S(A_n:A_{n-1})$. These calculations have obvious relevance for the problem of quantum retrodiction~\cite{HilleryKoch2016}, but we do not here derive optimal protocols to achieve this. 
 
 %%%%%%%%%%%%%%%%%%%%%%%%
 % Theorem - past info  %
 %%%%%%%%%%%%%%%%%%%%%%%%
 \begin{thm}\label{theorem-past-info}
 The information that the last device has about the previous device is reduced when that measurement is amplified. That is,
 \begin{equation}
    S(D_n : D_{n-1}) \le S(A_n : A_{n-1}) \, .
 \end{equation}
 \end{thm}

 \begin{proof}
 From the wavefunction~\eqref{eqn:FullWavefunction}, the density matrix for the last two ancillae in the measurement chain is
 \begin{equation} 
 \begin{split} \label{eqn:last-two-un-amplified}
        \rho(A_{n-1}A_n) = \!\!\!\!\! & \sum_{\substack{  x_{n-2}  x_{n-1} \\ x'_{n-1} x_n }}  \!\!\! q^{(n-2)}_{x_{n-2}} \, U^{(n-1)}_{\!x_{n-2}x_{n-1}}  U^{(n-1)*}_{\!x_{n-2}x'_{n-1}} \\
           & ~~ \times U^{(n)}_{\!x_{n-1} x_n}  U^{(n)*}_{\!x'_{n-1} x_n} \, |x_{n-1} x_n\ra\la x'_{n-1} x_n |.
 \end{split}
 \end{equation} 
 Amplification removes the off-diagonals of $\rho(A_{n-1}A_n)$ so that
 \begin{equation} \label{eqn:last-two-amplified}
        \rho(D_{n-1}D_n) = \sum_{x_{n-1}} P_{x_{n-1}} \,\, \rho(A_{n-1}A_n) \, P_{x_{n-1}},
 \end{equation} 
 where $P_{x_{n-1}} = |x_{n-1}\ra\la x_{n-1}|$ are projectors on the state of $A_{n-1}$. Note that, from~\eqref{eqn:last-two-un-amplified}, it is sufficient to amplify just the second-to-last measurement with $A_{n-1}$. Since the marginal entropies are unchanged by the amplification (Lemma~\ref{lemma-AiDi}), the amount of information before amplification, $S(A_n:A_{n-1})$, and after, $S(D_n:D_{n-1})$, is related by
  \begin{equation} 
        S(D_n\!:\!D_{n-1}) \!=\! S(A_n\!:\!A_{n-1}) + S(A_{n-1}A_n) - S(D_{n-1}D_n). 
 \end{equation} 
 In a similar fashion to the calculations in Theorem~\ref{theorem-last-info}, it is evident from~\eqref{eqn:last-two-amplified} that the joint entropy is increased, $S(D_{n-1}D_n) \ge S(A_{n-1}A_n)$. It follows that the information that the last device has about the device that preceded it in the measurement sequence is reduced:
 \begin{equation}
      S(D_n : D_{n-1}) \le S(A_n : A_{n-1}),  
 \end{equation} 
 with equality if and only if $\rho(A_{n-1}A_n)$ is already diagonal in the ancilla product basis.
 \end{proof}

 Using the case of qubits, we can show how amplification reduces the amount of information about past measurements. In this example, suppose that the last two measurements in the chain are each made at the relative angle $\pi/4$. As expected, the amplified density matrix~\eqref{eqn:last-two-amplified} becomes uncorrelated, $\rho(D_{n-1}D_n) = \frac12 \, \mathbbm{1}_{D_{n-1}} \otimes \frac12 \, \mathbbm{1}_{D_n}$, where $\mathbbm{1}$ is the $2\times2$ identity matrix, and the shared entropy vanishes $S(D_n : D_{n-1}) = 0$. In other words, the last detector has no information about the detector preceding it. In contrast, prior to amplification the density  matrix~\eqref{eqn:last-two-un-amplified} is coherent with joint entropy $S(A_{n-1}A_n) = 1 + S_{n-2}$. 
 %for these two angles, the eigenvalues of \rho(A_{n-1}A_n) are the doubly repeated values q/2 and (1-q)/2, where q = q^{(n-2)}_{x_{n-2} = 0}. 
 Therefore, the corresponding shared entropy is nonzero, $S(A_n : A_{n-1}) = 1 - S_{n-2} = 1 - H[q^{(n-2)}]$, revealing that information about the previous measurement survives the sequential $\pi/4$ measurements (as long as $A_{n-1}$ is not amplified).
 
 The calculations described above can be extended to include the information that the last device has about {\em all} previous devices in the measurement chain. We claim in Theorem~\ref{theorem-chain} that the amplification process reduces this information by a specific minimum (calculable) amount. To prove this statement, we make use of Theorem~\ref{theorem-cq-prepared}, where we showed that the joint entropy of all quantum ancillae that measured a prepared quantum system is simply equal to the entropy of last ancilla in the unamplified chain.

 %%%%%%%%%%%%%%%%%%%%%%%%
 % Theorem - chain      %
 %%%%%%%%%%%%%%%%%%%%%%%%
 \begin{thm}\label{theorem-chain}
 For $n$ consecutive measurements of a quantum system, the information that the last device has about all previous measurements is reduced by amplification by at least an amount $\Sigma_n$:
 \begin{equation} 
        S(D_n:D_{n-1} \ldots D_1) \le S(A_n:A_{n-1} \ldots A_1) - \Sigma_n\;,
 \end{equation} 
 where $\Sigma_n = S(A_{n-1}|A_n) \ge 0$ is a non-negative conditional entropy that quantifies the uncertainty about the prior measurement given the last. 
 \end{thm}

 \begin{proof}
 We begin by recognizing that the amplified mutual entropy $S(D_n:D_{n-1} \ldots D_1)$ for the full measurement chain is equal to $S(D_n:D_{n-1})$ by the Markov property (see Theorem~\ref{theorem-markov}). Then, by Theorem~\ref{theorem-past-info}, we can place an upper bound on the amplified information
 \begin{equation}\label{eqn:upper-bound}
     S(D_n:D_{n-1} \ldots D_1) = S(D_n:D_{n-1}) \le S(A_n:A_{n-1}),
 \end{equation}
 where $S(A_n:A_{n-1})$ is the mutual entropy before amplifying the measurement. Next, we will relate $S(A_n:A_{n-1})$ to $S(A_n:A_{n-1} \ldots A_1)$. From Theorem~\ref{theorem-cq-prepared}, the latter quantity can be written simply as
 \begin{equation}
     S(A_n:A_{n-1} \ldots A_1) = S_{n-1}, 
 \end{equation}
 so that with the definition of $S(A_n:A_{n-1})$, we come to
 \begin{equation}\label{eqn:eta}
    S(A_n:A_{n-1} \ldots A_1) = S(A_n:A_{n-1}) + \Sigma_n,
 \end{equation}
 where $\Sigma_n = S(A_{n-1}|A_n)$ represents the information gained by conditioning on all previous measurements. Inserting~\eqref{eqn:eta} into the inequality~\eqref{eqn:upper-bound}, we come to
 \begin{equation}
    S(D_n:D_{n-1} \ldots D_1) \le S(A_n:A_{n-1} \ldots A_1) - \Sigma_n.
 \end{equation}

 The information is reduced as long as $\Sigma_n \ge 0$. To show this, we recall the joint density matrix~\eqref{eqn:last-two-un-amplified} for $A_{n-1}$ and $A_n$. This state can be written as a classical-quantum state
 \begin{equation}
        \rho(A_{n-1}A_n) = \sum_{x_n} q^{(n)}_{x_n} \,\, \rho_{x_n} \otimes |x_n\ra\la x_n|,
 \end{equation}
 where 
 \begin{equation}
    q^{(n)}_{x_n} \,\, \rho_{x_n} \!= \!\!\sum_{x_{n-2}} q^{(n-2)}_{x_{n-2}} \,\, p^{(n-2,n)}_{x_{n-2}x_n} \,\, |\phi_{x_{n-2}x_n} \ra\la \phi_{x_{n-2}x_n}|,
 \end{equation}
 and the non-orthogonal states $|\phi_{x_{n-2}x_n} \ra$ were previously defined in Eq.~\eqref{diagonalize-unprepared-general}. In this block-diagonal form, the entropy is
 \begin{equation}
        S(A_{n-1}A_n) = S_n + \sum_{x_n} q^{(n)}_{x_n} \, S(\rho_{x_n}),
 \end{equation}
 so that the quantity of interest, $\Sigma_n$, can be written as
 \begin{equation}
     \Sigma_n = S(A_{n-1}|A_n) = \sum_{x_n} q^{(n)}_{x_n} \, S(\rho_{x_n}) \ge 0.
 \end{equation}
 This quantity is clearly non-negative since both $q^{(n)}_{x_n} \ge 0$ and $S(\rho_{x_n}) \ge 0 \,\, \forall \,\, x_n$. Therefore, with $\Sigma_n \ge 0$, we find that the information is indeed reduced by the amplification process, and by at least an amount equal to $\Sigma_n$.
 \end{proof}
 
 Continuing with our qubit example that followed Theorem~\ref{theorem-past-info}, if the last two measurements were each made at the relative angle $\pi/4$, the ancilla $A_n$ has 1 bit of information about the joint state of all previous ancillae. That is, $S(A_n:A_{n-1} \ldots A_1) = 1$ bit, while the amplifying detector $D_n$ has no information at all, $S(D_n:D_{n-1} \ldots D_1) = 0$.

 %%%%%%%%%%%%%%%%%%%%%%%%
 %  Corollary - chain   %
 %%%%%%%%%%%%%%%%%%%%%%%%
 \begin{cor}\label{corollary-chain}
 Amplifying the measurement chain increases the entropy of the last device, when conditioned on all previous devices, by at least an amount $\Sigma_n$:
 \begin{equation}
    S(D_n|D_{n-1} \ldots D_1) \ge S(A_n|A_{n-1} \ldots A_1) + \Sigma_n.
 \end{equation}
 \end{cor}
 \begin{proof}
 By definition, the mutual entropy and conditional entropy are related by
 \begin{equation}
    S(D_n:D_{n-1} \ldots D_1) = S_n -  S(D_n|D_{n-1} \ldots D_1),    
 \end{equation}
 which, from Theorem~\ref{theorem-chain}, is bounded from above by $S(A_n:A_{n-1} \ldots A_1) - \Sigma_n = S_n -  S(A_n|A_{n-1} \ldots A_1) - \Sigma_n$. Therefore,
 \begin{equation}
    S(D_n|D_{n-1} \ldots D_1) \ge S(A_n|A_{n-1} \ldots A_1) + \Sigma_n,
 \end{equation}
 and the uncertainty in the last measurement is increased by at least an amount $\Sigma_n$.
 \end{proof}

\begin{figure} % Fig. 14
    \centering
    \includegraphics[width=\linewidth]{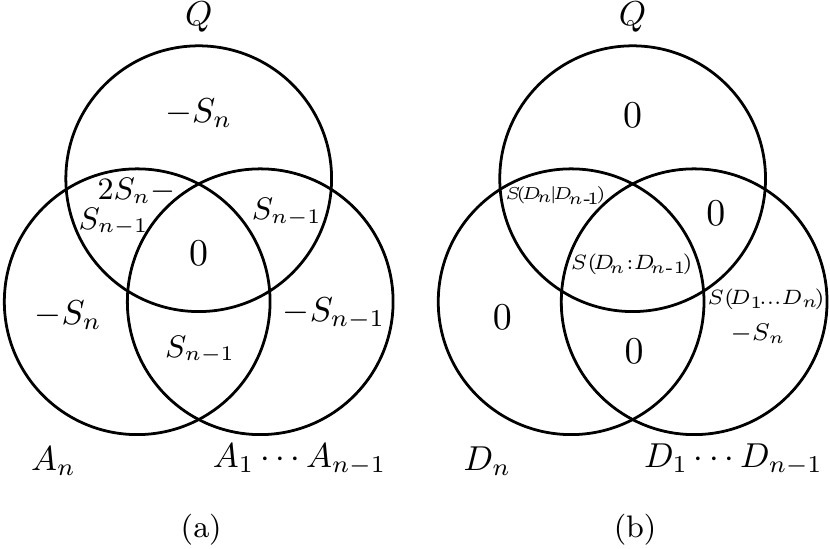} %venn_summary-crop.pdf
    \caption{Entropy Venn diagrams (a) before amplification with $n$ ancillae that consecutively measured a quantum system $Q$, and (b) after amplification with macroscopic detectors.} 
    \label{fig:venn-summary} 
\end{figure}

 This section quantified a number of unsurprising, but nevertheless important results: amplifying measurements reduces information, and increases uncertainty. The key quantity that characterizes the difference between unamplified and amplified chains is $\Sigma_n$, which quantifies how much we do {\em not} know about the state preparation, $A_{n-1}$, given the state determination, $A_n$. Depending on the relative state between $A_{n-1}$ and $A_n$, we may know nothing ($\Sigma_n=1$), or everything ($\Sigma_n=0$). We summarize the results presented in this section with the entropy Venn diagrams in Fig.~\ref{fig:venn-summary}.

\section{Applications of consecutive quantum measurements} 
\label{sec:six}
 The formalism developed in this paper can be directly applied to several interesting situations. Here, we focus specifically on the double-slit experiment, the quantum Zeno effect, and quantum state preparation.

 \subsection{The double-slit experiment}
 Suppose a photon in the state
 \begin{equation}
     |\Psi\ra = |h\ra_P \otimes |\psi\ra_Q,
 \end{equation}
 is incident on a double-slit apparatus. Initially, it has polarization (denoted $P$) degree of freedom $h$, and spatial (denoted $Q$) degree of freedom $\psi$. Once past the slits, its spatial state evolves to the superposition
 \begin{equation}
 \label{doubleslit-unentangled}
     |\Psi\ra = |h\ra_P \otimes \frac{|\psi_1\ra_Q + |\psi_2\ra_Q}{\sqrt2},
 \end{equation}
 where $|\psi_j\ra_Q$ is the state corresponding to the photon passing through slit $j$. The photon is then detected by a CCD camera $D_X$, which acts as an interference screen. This interaction can be modeled as a von Neumann measurement of the spatial states by the screen. Expanding the spatial states of the photon in terms of the position basis of the screen yields
 \begin{equation}
    |\psi_j\ra_Q = \sum_{x=1}^n \psi_j(x) \, |x\ra_Q,
 \end{equation} 
 where $j=0,1$ labels each slit. The states $|x\ra$ can be discretized into $n$ distinct locations according to 
 \begin{equation} 
 \begin{split} 
    |x=1\ra &= |100 \ldots 0\ra,\\
    |x=2\ra &= |010 \ldots 0\ra,\\
            &\vdots\\
    |x=n\ra &= |0 \ldots 001\ra, 
 \end{split} 
 \end{equation}
 which denote the location $x$ at which a photon is detected by $D_X$. Inserting this basis into the expression~\eqref{doubleslit-unentangled} and performing the measurement of $Q$ with $D_X$ (which starts in the initial state $|x=0\ra = |0 \ldots 0\ra$), we come to
 \begin{equation}
    |\Psi'\ra = |h\ra_P \otimes \sum_{x=1}^n \, \frac{ \psi_1(x) + \psi_2(x) }{\sqrt2} \,\, |xx\ra_{QD_X}.
 \end{equation}
 Tracing out the photon states, the density matrix describing the screen is
 \begin{equation}
     \rho(D_X) = \frac12 \sum_{x=1}^n \, \big| \psi_1(x) + \psi_2(x) \big|^2 \,\, |x\ra_{D_X\!}\la x|,
 \end{equation}
 where the probability to detect the photon at a position $x$ is a coherent superposition of probability amplitudes and leads to the standard double-slit interference pattern.
 
 We can extend this description to the case of multiple measurements in the context of the quantum eraser experiment. We first tag the photon's path in order to obtain information about through which slit it passed. In practice, we can implement the tagging operation by placing different wave plates in front of each slit. As a simple example, we assume the tagging takes the form of a controlled-\textsc{not} operation so that horizontal polarization, $h$, is converted to vertical polarization, $v$, if the photon traverses the second slit. Thus, instead of~\eqref{doubleslit-unentangled}, the polarization, $P$, and spatial, $Q$, degrees of freedom are now entangled,
 \begin{equation}
 \label{doubleslit-entangled}
     |\Psi\ra = \frac{|h\ra_P \otimes |\psi_1\ra_Q + |v\ra_P \otimes |\psi_2\ra_Q}{\sqrt2}.
 \end{equation}
 
 Of course, the entanglement in~\eqref{doubleslit-entangled} destroys the interference pattern on the screen. The fringes can be restored by measuring the photon's polarization with a detector $D_P$ in a rotated basis, before the photon hits the screen. Rewriting the polarization states in the new basis, $|0\ra$ and $|1\ra$, which are rotated by an angle $\theta$ with respect to $|h\ra$ and $|v\ra$ [see~\eqref{rotation}],
  \begin{align}
     |v\ra_P & = U_{00} \, |0\ra_P + U_{01} \, |1\ra_P, \\
     |h\ra_P & = U_{10} \, |0\ra_P + U_{11} \, |1\ra_P,
 \end{align}
 and measuring $P$ with the detector $D_P$ yields
 \begin{equation}
 \label{doubleslit-entangled-erased}
 \begin{split} 
     |\Psi'\ra & = \frac{1}{\sqrt2} \Big( U_{10} \, |\psi_1\ra_Q \!+ U_{00} \, |\psi_2\ra_Q \Big) \!\otimes |00\ra_{PD_P} \\
     & \,+ \frac{1}{\sqrt2} \Big( U_{11} \, |\psi_1\ra_Q \!+ U_{01} \, |\psi_2\ra_Q \Big) \!\otimes |11\ra_{PD_P}.
 \end{split}
 \end{equation}

 The angle at which we measure the polarization determines the coherence of the spatial states $Q$, which is reflected in the visibility of the recovered interference patterns. Repeating the measurement with the screen, \eqref{doubleslit-entangled-erased} becomes
 \begin{equation}
 \label{doubleslit-entangled-screen}
 \begin{split} 
     |\Psi''\ra & = \frac{1}{\sqrt2}\, \sum_{x=1}^n \!\Bigg[  \Big( U_{10}  \psi_1(x) \!+ U_{00}  \psi_2(x) \Big) |00\ra_{PD_P} \\
     & ~~~~~ + \!\Big( U_{11}  \psi_1(x) \!+ U_{01} \psi_2(x) \Big)  |11\ra_{PD_P}  \Bigg] \! \otimes |xx\ra_{QD_X}.
 \end{split}
 \end{equation}  
 The density matrix for the screen is, as expected, still completely mixed, and describes two intensity peaks on the screen. However, an interference pattern can be extracted if we condition on the outcome of the polarization measurement. That is,
 \begin{equation}
     \rho(QD_P) = \frac12 \, \sum_{i=0}^1 \,\, \rho_Q^{(i)} \otimes |i\ra_{D_P}\la i|,
 \end{equation}
 where
 \begin{equation}
     \rho_Q^{(i)} = \sum_{x=1}^n \, \big| U_{1i} \,  \psi_1(x) + U_{0i} \, \psi_2(x)  \big|^2 \,\, |x\ra_{D_X}\la x|,
 \end{equation}
 is the state of $Q$, given that the polarization measurement yielded the outcome $i$. The probability distribution of this state is a coherent sum of amplitudes and describes an interference pattern with a visibility that is controlled by the measurement angle, $\theta$. In particular, measuring at $\theta = 0$ leads to no interference, while $\theta = \pi/4$ recovers the standard fringe (or anti-fringe) pattern. We refer the reader to Ref.~\cite{GlickAdami-bell-2017} for a detailed information-theoretic analysis of the Bell-state quantum eraser experiment, where the degree of erasure is controlled by an entangled photon partner, even after the original photon has hit the screen.

 \subsection{Quantum Zeno and anti-Zeno effects} 
 \label{sec:zeno}
  In this section, we derive the standard results of the quantum Zeno and anti-Zeno effects in the context of unitary consecutive measurements. Instead of a time-varying quantum state controlled by quantum measurements in the same basis, we can equivalently study a static quantum state consecutively measured by quantum detectors whose basis changes in time.

 For the Zeno effect~\cite{HomeWhitaker1997,Peres1995}, we assume that an initial quantum two-state system is in the state $\rho(Q) = p |0\ra\la 0| + (1-p) |1\ra\la1|$, with arbitrary $p$, which was prepared by a measurement with detector $D_1$. It is then subsequently measured by detectors $D_2\ldots D_n$, each at an angle $\pi/(4n)$ relative to the previous detector, completing a full $\pi/4$ rotation after $n$ observations. The density matrix for the preparation with the first detector $D_1$
is $\rho(D_1)=p|0\ra\la0|+(1-p)|1\ra\la1|$, which has an entropy $S(D_1) = -p\log_2 p - (1-p)\log_2 (1-p)$. The density matrix for the second detector, expressed in a different basis that is rotated with the unitary matrix $U$, is
 \begin{equation} 
    \rho(D_2)=\sum_j\left(p\,|U_{0j}|^2+(1-p)\,|U_{1j}|^2\right)|j\ra\la j|\,,
 \end{equation}
 where the unitary matrix is given by
 \begin{equation}
    U = \left(\begin{array}{cc}\cos(\frac\pi{4n})&-\sin(\frac\pi{4n})\\
         \sin(\frac\pi{4n})& ~~ \cos(\frac\pi{4n})\end{array}\right)\,.
 \end{equation}
 The entropy of the second detector is $S(D_2) = - q \log_2 q - (1-q)\log_2 (1-q)$ with $q=1/2+(p-1/2)\cos\left(\frac{\pi}{2n}\right)$, the probability to observe the state $|0\ra$ for the second measurement. Figure~\ref{fig:zeno} shows the detector entropies $S(D_1)$, $S(D_2)$, and $S(D_3)$ for measurements $D_2$ and $D_3$  after the preparation with $D_1$. 
 \begin{figure} % Fig. 15
    \centering
    \includegraphics[width=\linewidth]{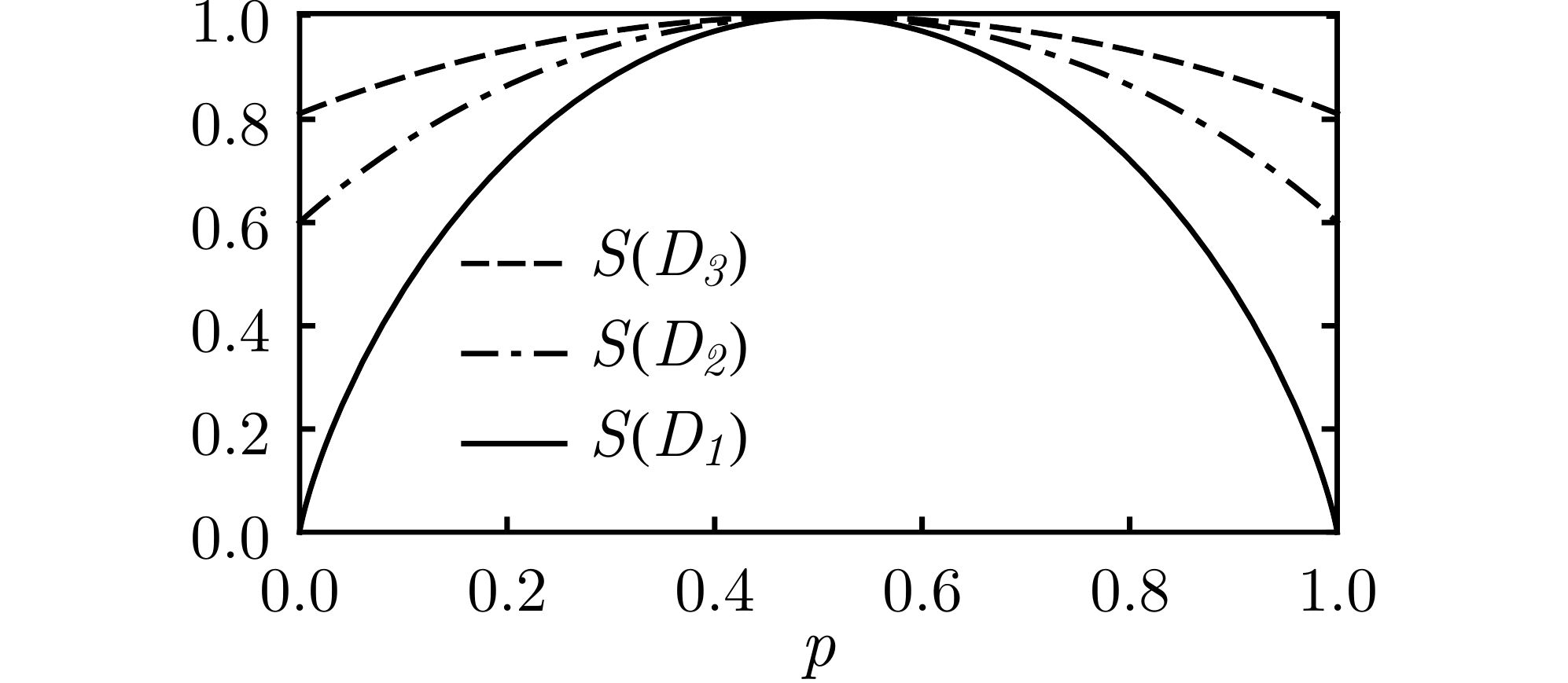} %zeno-margin.pdf
    \caption{The detector entropies for two consecutive measurements $D_2$ and $D_3$, after the preparation with $D_1$, of the quantum state $\rho(Q) = p \, |0\ra\la 0| + (1-p) \, |1\ra\la1|$. Each detector is at an angle of $\pi/8$ relative to the previous detector.}
    \label{fig:zeno}
 \end{figure}

 In general, following the preparation, the probability $q^{(n)}$ to observe the state $|0\ra$ after $n$ measurements is
 \begin{equation}
    q^{(n)}=\frac12+\left(p-\frac12\right)\cos^{\,n}\!\left(\frac{\pi}{2n}\right)\to p\ \ {\rm as}\ n\to\infty\;.
 \end{equation}
 In other words, the density matrix of the $n$th detector is equal to that of the preparation with $D_1$. For polarization measurements for example, this results in perfect transmission of the initially polarized beam even though the $n$ detectors rotate the plane of polarization by 45 degrees~\cite{Kofmanetal2001}.

 The anti-Zeno effect is often described as the complete destruction of a quantum state due to incoherent consecutive measurements~\cite{KaulakysGontis1997,LewensteinRzazewski2000,Luis2003}. In the present language, this corresponds to the randomization of a given (prepared) quantum state after consecutive measurements at random angles with respect to the initial state. We begin again with the prepared state $\rho(Q)=p|0\ra\la0|+(1-p)|1\ra\la1|$, but now observe it consecutively using measurement devices $D_k$ at angles $\theta_k$ drawn from a uniform distribution on the interval $[0,\pi/4]$. The probability to observe $Q$ in state $|0\ra$ after $n$ measurements with random phases is now
 \begin{equation}
    q^{(n)}=\frac12+\left(p-\frac12\right)\Pi_{k=1}^n\cos(2\theta_k)\;.
 \end{equation}
 In order to obtain the most likely state probability for random dephasing, we calculate the expectation value,
 \begin{equation}
    E\left[\Pi_{k=1}^n\cos(2\theta_k)\right]=\Pi_{k=1}^n E\left[\cos(2\theta_k)\right]=\left(\frac2\pi\right)^{\!n} \!,
 \end{equation}
 so that $E\left[q^{(n)}\right] \to 1/2$ as $n\to \infty$. Thus, any quantum state is randomized via consecutive quantum projective measurements in random bases. A similar result was derived for the dephasing of photon polarization in Ref.~\cite{Kofmanetal2001}.

 \subsection{Preparing quantum states}
 For our final application, we discuss how to prepare quantum states by considering consecutive measurements on unprepared quantum states. Suppose a quantum system is prepared in the known state
 \begin{equation} \label{prep}
    \rho(Q)=\sum_{x=1}^d p_x \, |\wt{x}\ra\la \wt{x}|\;,
 \end{equation}
 which we already wrote in the basis of the ancilla that will perform the first measurement after the preparation. We can always prepare a state like~\eqref{prep} by measuring an unprepared quantum state~\eqref{QR}, with the pair $A_1D_1$ in a given, but arbitrary basis. Then, a second measurement with $A_2D_2$ at a relative angle $\theta_2$ gives rise to the state
 \begin{equation}
     |QRA_1D_1A_2D_2\ra = \frac{1}{\sqrt d} \sum_{x_1 x_2} U^{(2)}_{x_1 x_2} ~ |\wt{x}_2 \,\, x_1 \, x_1 x_1 \, x_2 x_2 \ra.
 \end{equation}
 From this we can compute the operator~\cite{CerfAdami1998} describing the state of the quantum system, conditional on the state of the first detector, $D_1$,
 \begin{equation} 
   \begin{split} 
    \rho(Q|D_1) & = \rho(Q D_1) \, \Big(\rho(D_1)^{-1} \otimes \mathbbm{1}_Q \Big) \\
       & = \sum_{x_1} \rho^{x_1}_Q \otimes |x_1\ra\la x_1|\,,
  \end{split}
 \end{equation}
 where $\rho(D_1)^{-1}$ is the inverse of the density matrix. The density matrix $\rho^{x_1}_Q$ is the prepared state~\eqref{prep} of the quantum system, given that the outcome $x_1$ was observed in the first measurement,
 \begin{equation} \label{Qprojected}
    \rho^{x_1}_Q = \frac{{\rm Tr}_{D_1} \! \left[P_{x_1}\rho(QD_1)P_{x_1}^\dagger\right]}{ {\rm Tr}_{QD_1} \!\! \left[P_{x_1}\rho(QD_1)P_{x_1}^\dagger\right] }  =  \sum_{x_2}  |U^{(2)}_{x_1 x_2}|^2 \,  |\wt{x}_2\ra\la \wt{x}_2|.
 \end{equation}
 Here, $P_{x_1} = |x_1\ra\la x_1|$ are projectors on the state of detector $D_1$. If we choose for the quantum state preparation the outcome $x_1=0$, for example, then $p_{x_2}=|U^{(2)}_{0x_2}|^2$ provides the probability distribution for the quantum system, and we arrive at the desired prepared state~\eqref{prep} from~\eqref{Qprojected}. 
 
 The purification of~\eqref{prep} in terms of the basis of ancilla $A_2$ is 
 \begin{equation}\label{QD2}
    |QA_2\ra=\sum_{x_2}\sqrt{p_{x_2}} \,\, |\wt{x}_2\ra |x_2\ra, 
 \end{equation}
 which is an entangled state with the marginal entropies $S(Q) \!=\! S(A_2) \!=\! H[p]$. If we rename $A_2$ to $A_1$, then expression~\eqref{QD2} is equivalent to~\eqref{QA_stageI}. Equipped with this state preparation, we can now make the usual consecutive (amplified or unamplified) measurements of $Q$ with $A_2D_2, \, A_3D_3,$ etc.

 \section{Conclusions}
 \label{sec:seven}
 Conventional wisdom in quantum mechanics dictates that the measurement process ``collapses'' the state of a quantum system so that the probability that a particular detector fires depends only on the state preparation and the measurement chosen. This assertion can be tested by considering sequences of measurements of the same quantum system. If a ``memory" of the first measurement (the state preparation) persists beyond the second measurement, then a reduction of the wave packet can be ruled out. We discussed two classes of quantum measurement: those performed within a closed system where every part of a measurement device (every qudit of the pointer) is under control, and those performed within an open system, where part of the pointer variable is ignored. We found that sequences of quantum measurements in closed systems are non-Markovian (retaining the memory of past measurements) while sequences of open-system measurements obey the Markov property. In the latter case, the probability distribution of future measurement results only depends on the state preparation and the measurement chosen. It is clear from our construction that the Markovian measurements are a special case of the non-Markovian ones, and that the loss of memory is not a fundamental property of quantum measurements, but is merely a consequence of the loss of quantum information when tracing over degrees of freedom that participated in the measurement. We quantified this loss by calculating the amount of information lost when observing coherent quantum detectors using incoherent devices. 
 
 We have found that the entropy of coherent chains of measurements is entirely determined by the entropy at the boundary of the chain, namely the entropy of the state {\em preparation} (the first measurement in the chain) and the last measurement. (If the chain is started on a known state, then the entropy of the chain is contained in the last measurement only). This property is a direct consequence of the unitarity of quantum measurements, and signifies that any quantum measurement outcome is constrained by its immediate past and its immediate future. It has not escaped our attention that this property of quantum chains is reminiscent of the holographic principle, which posits that the description of a system can be encoded entirely on its boundary alone. Because the holographic principle is often thought to have its origin in an information-theoretic description of space-time~\cite{Wheeler1990}, it is perhaps not surprising that an information-theoretic analysis of chains of measurements would yield precisely such an outcome. In particular, it is not too hard to imagine that the past-future relationship that consecutive quantum measurements entail create precisely the partial order required for the ``causal sets" program for quantum gravity~\cite{Bombellietal1987}. Of course, to recover space-time from sets of measurements we would need to consider not just sequential measurements on the same system, but multiple parallel chains that are entangled with each other, creating a network rather than a chain (we have recently shown that the unitary formalism deployed here can be extended to parallel measurements when discussing the Bell-state quantum eraser~\cite{GlickAdami-bell-2017}). In that respect, the network of quantum measurements is more akin to van Raamsdonk's~\cite{VanRaamsdonk2010} tensor networks, created using entangling and disentangling operations (see also~\cite{Pastawskietal2015}). Incidentally, the present formalism implies the existence of a disentangling operation that ``undoes" quantum measurements, and that can serve as a powerful primitive for controlling quantum entanglement~\cite{GlickAdami2017}.
 
 Using a quantum-information-theoretic approach, we have argued that a collapse picture makes predictions that differ from those of the unitary (relative state) approach if multiple consecutive non-Markovian measurements are considered. Should future experiments corroborate the manifestly unitary formulation we have outlined, such results would further support the notion of the reality of the quantum state~\cite{Pusey2012} and that the wavefunction is not merely a bookkeeping device that summarizes an observer's knowledge about the system~\cite{Englert2013,Fuchsetal2014}. We hope that moving discussions about the nature of quantum reality from philosophy into the empirical realm will ultimately lead to a more complete (and satisfying) understanding of quantum physics.

 \begin{acknowledgments}
   CA would like to thank Jeff Lundeen and his group for discussions. Financial support by a Michigan State University fellowship to JRG is gratefully acknowledged.
 \end{acknowledgments}

 %\appendix*
 
 \bibliography{quant}
 
\end{document}